\documentclass[12pt]{amsart}
\reversemarginpar
\usepackage{amsmath}
\usepackage{amssymb}
\usepackage{graphicx}
\usepackage{graphicx}
\usepackage{dcolumn}
\usepackage{bm}
\usepackage{amsfonts}
\usepackage{latexsym}
\usepackage{pdfsync}

\begin{document}
\newcommand{\commentout}[1]{}

\newcommand{\nwc}{\newcommand}
\newcommand{\bz}{{\mathbf z}}
\newcommand{\sqk}{\sqrt{\ks}}
\newcommand{\sqkone}{\sqrt{|\ks_1|}}
\newcommand{\sqktwo}{\sqrt{|\ks_2|}}
\newcommand{\invsqkone}{|\ks_1|^{-1/2}}
\newcommand{\invsqktwo}{|\ks_2|^{-1/2}}
\newcommand{\partz}{\frac{\partial}{\partial z}}
\newcommand{\grady}{\nabla_{\by}}
\newcommand{\gradp}{\nabla_{\bp}}
\newcommand{\gradx}{\nabla_{\bx}}
\newcommand{\invf}{\cF^{-1}_2}
\newcommand{\myphi}{\Phi_{(\eta,\rho)}}
\newcommand{\minrg}{|\min{(\rho,\gamma^{-1})}|}
\newcommand{\al}{\alpha}
\newcommand{\xvec}{\vec{\mathbf x}}
\newcommand{\kvec}{{\vec{\mathbf k}}}
\newcommand{\lt}{\left}
\newcommand{\ksq}{\sqrt{\ks}}
\newcommand{\rt}{\right}
\nwc{\bG}{{\bf G}}
\newcommand{\ga}{\gamma}
\newcommand{\vas}{\varepsilon}
\newcommand{\lan}{\left\langle}
\newcommand{\ran}{\right\rangle}
\newcommand{\tvas}{{W_z^\vas}}
\newcommand{\psiep}{{W_z^\vas}}
\newcommand{\wep}{{W^\vas}}
\newcommand{\weptil}{{\tilde{W}^\vas}}
\newcommand{\wepz}{{W_z^\vas}}
\newcommand{\weps}{{W_s^\ep}}
\newcommand{\wepsp}{{W_s^{\ep'}}}
\newcommand{\wepzp}{{W_z^{\vas'}}}
\newcommand{\wepztil}{{\tilde{W}_z^\vas}}
\newcommand{\vvas}{{\tilde{\ml L}_z^\vas}}
\newcommand{\veptil}{{\tilde{\ml L}_z^\vas}}
\newcommand{\vep}{{{ V}_z^\vas}}
\newcommand{\cvc}{{{\ml L}^{\ep*}_z}}
\newcommand{\cvcp}{{{\ml L}^{\ep*'}_z}}
\newcommand{\cvp}{{{\ml L}^{\ep*'}_z}}
\newcommand{\cvtil}{{\tilde{\ml L}^{\ep*}_z}}
\newcommand{\cvtilp}{{\tilde{\ml L}^{\ep*'}_z}}
\newcommand{\vtil}{{\tilde{V}^\ep_z}}
\newcommand{\ktil}{\tilde{K}}
\newcommand{\n}{\nabla}
\newcommand{\tkappa}{\tilde\kappa}
\newcommand{\ks}{{\omega}}
\newcommand{\bx}{\mb x}
\newcommand{\br}{\mb r}
\nwc{\bH}{{\mb H}}
\newcommand{\bu}{\mathbf u}
\nwc{\bxp}{{{\mathbf x}}}
\nwc{\byp}{{{\mathbf y}}}
\newcommand{\bD}{\mathbf D}
\nwc{\bh}{\mathbf h}
\newcommand{\bB}{\mathbf B}
\newcommand{\bC}{\mathbf C}
\nwc{\cO}{\mathcal  O}
\newcommand{\bp}{\mathbf p}
\newcommand{\bq}{\mathbf q}
\newcommand{\by}{\mathbf y}
\nwc{\bP}{\mathbf P}
\nwc{\bs}{\mathbf s}
\nwc{\bX}{\mathbf X}
\newcommand{\pdg}{\bp\cdot\nabla}
\newcommand{\pdgx}{\bp\cdot\nabla_\bx}
\newcommand{\one}{1\hspace{-4.4pt}1}
\newcommand{\corr}{r_{\eta,\rho}}
\newcommand{\rinf}{r_{\eta,\infty}}
\newcommand{\rzero}{r_{0,\rho}}
\newcommand{\rzeroinf}{r_{0,\infty}}
\nwc{\om}{\omega}
\nwc{\Gp}{{G_{\rm par}}}
\nwc{\nwt}{\newtheorem}
\nwc{\xp}{{x^{\perp}}}
\nwc{\yp}{{y^{\perp}}}
\nwt{remark}{Remark}
\nwt{definition}{Definition} 
\nwc{\bd}{{\mb d}}
\nwc{\ba}{{\mb a}}
\nwc{\bal}{\begin{align}}
\nwc{\be}{\begin{equation}}
\nwc{\ben}{\begin{equation*}}
\nwc{\bea}{\begin{eqnarray}}
\nwc{\beq}{\begin{eqnarray}}
\nwc{\bean}{\begin{eqnarray*}}
\nwc{\beqn}{\begin{eqnarray*}}
\nwc{\beqast}{\begin{eqnarray*}}

\nwc{\eal}{\end{align}}
\nwc{\ee}{\end{equation}}
\nwc{\een}{\end{equation*}}
\nwc{\eea}{\end{eqnarray}}
\nwc{\eeq}{\end{eqnarray}}
\nwc{\eean}{\end{eqnarray*}}
\nwc{\eeqn}{\end{eqnarray*}}
\nwc{\eeqast}{\end{eqnarray*}}

\nwc{\ep}{\varepsilon}
\nwc{\eps}{\varepsilon}
\nwc{\ept}{\epsilon}
\nwc{\vrho}{\varrho}
\nwc{\orho}{\bar\varrho}
\nwc{\ou}{\bar u}
\nwc{\vpsi}{\varpsi}
\nwc{\lamb}{\lambda}
\nwc{\Var}{{\rm Var}}

\nwt{cor}{Corollary}
\nwt{proposition}{Proposition}
\nwt{corollary}{Corollary}
\nwt{theorem}{Theorem}
\nwt{summary}{Summary}
\nwt{lemma}{Lemma}
\nwc{\nn}{\nonumber}
\nwc{\mf}{\mathbf}
\nwc{\mb}{\mathbf}
\nwc{\ml}{\mathcal}
\nwc{\bj}{{\mb j}}
\nwc{\bA}{{\mb \Phi}}
\nwc{\IA}{\mathbb{A}} 
\nwc{\bi}{\mathbf i}
\nwc{\bo}{\mathbf o}
\nwc{\IB}{\mathbb{B}}
\nwc{\IC}{\mathbb{C}} 
\nwc{\ID}{\mathbb{D}} 
\nwc{\IM}{\mathbb{M}} 
\nwc{\IP}{\mathbb{P}} 
\nwc{\bI}{\mathbf{I}} 
\nwc{\IE}{\mathbb{E}} 
\nwc{\IF}{\mathbb{F}} 
\nwc{\IG}{\mathbb{G}} 
\nwc{\IN}{\mathbb{N}} 
\nwc{\IQ}{\mathbb{Q}} 
\nwc{\IR}{\mathbb{R}} 
\nwc{\IT}{\mathbb{T}} 
\nwc{\IZ}{\mathbb{Z}} 
\nwc{\II}{\mathbb{I}} 

\nwc{\cE}{{\ml E}}
\nwc{\cP}{{\ml P}}
\nwc{\cQ}{{\ml Q}}
\nwc{\cL}{{\ml L}}
\nwc{\cX}{{\ml X}}
\nwc{\cW}{{\ml W}}
\nwc{\cZ}{{\ml Z}}
\nwc{\cR}{{\ml R}}
\nwc{\cV}{{\ml V}}
\nwc{\cT}{{\ml T}}
\nwc{\crV}{{\ml L}_{(\delta,\rho)}}
\nwc{\cC}{{\ml C}}
\nwc{\cA}{{\ml A}}
\nwc{\cK}{{\ml K}}
\nwc{\cB}{{\ml B}}
\nwc{\cD}{{\ml D}}
\nwc{\cF}{{\ml F}}
\nwc{\cS}{{\ml S}}
\nwc{\cM}{{\ml M}}
\nwc{\cG}{{\ml G}}
\nwc{\cH}{{\ml H}}
\nwc{\cN}{{\ml N}}
\nwc{\bk}{{\mb k}}
\nwc{\bT}{{\mb T}}
\nwc{\cbz}{\overline{\cB}_z}
\nwc{\supp}{{\hbox{\rm supp}}}
\nwc{\fR}{\mathfrak{R}}
\nwc{\bY}{\mathbf Y}
\newcommand{\mbr}{\mb r}
\nwc{\pft}{\cF^{-1}_2}
\nwc{\bU}{{\mb U}}
\nwc{\bPhi}{{\mb \Phi}}
\nwc{\bPsi}{{\mb \Psi}}
\title{Exact Localization and Superresolution  with Noisy Data and Random Illumination
}
\author{Albert C.  Fannjiang}
\email{
fannjiang@math.ucdavis.edu}
  \thanks{The research is partially supported by
the NSF grant DMS - 0908535.}

      \address{
   Department of Mathematics,
    University of California, Davis, CA 95616-8633}
   
       \begin{abstract}
This paper studies the problem of exact localization of
multiple objects with noisy data. The crux of the proposed
approach consists of random illumination. Two recovery methods
are analyzed: the Lasso and the One-Step Thresholding (OST).  

For independent
random probes, it is shown that both recovery methods  can localize exactly $s=\cO(m)$, up to a logarithmic factor,  objects where
$m$ is the number of data. Moreover, when the number of random probes is large  the Lasso with random illumination has
a performance guarantee for superresolution, beating the
Rayleigh resolution limit. Numerical evidence confirms
the predictions and  indicates that the performance of the Lasso is superior
to that of the OST for the proposed set-up with random  illumination. 
       \end{abstract}
       
       \maketitle

\section{Introduction}

Two-point resolution is a standard criterion for
evaluation of imaging systems, i.e. the ability of the imaging system
to distinguish two closely located point objects. The
smallest resolvable distance $\ell$ between two objects, called  the (two-point) resolution length, is then defined as a metric of
the resolving power 
of the imaging system. Let $A$ be the aperture of the imaging
system, $z_0$ the distance to the objects and $\lambda$
the wavelength. The classical Rayleigh resolution criterion then
states
\beq\label{ray}
{A\ell\over z_0\lambda}=\cO(1)
\eeq
where there is some arbitrariness in the constant depending
on the precise definition of minimum  resolvable length $\ell$.

For noisy data, such a  criterion is more difficult 
to apply as  determination of $\ell$ becomes a statistical problem. 
One option would be to formulate the two-point resolution
problem as a statistical-hypothesis-testing problem
(one  versus two objects), see \cite{FS,SM} and references
therein. 
However, it is
cumbersome to generalize this approach to multiple point objects. 

 In this paper we first study the resolution issue  from
 the perspective of exact,  simultaneous  localization of multiple point objects. We evaluate an imaging method by saying that
 it can exactly localize $s$ (sparsity) randomly distributed point objects
mutually  separated by a minimum distance $\ell$ with high probability. In addition to reconsidering  the issue of resolution, 
 we seek an approach 
that can recover a high number $s=\cO(m)$ objects where
$m$ is the number of data, 
with resolution $\ell$ far below what is dictated by
the Rayleigh resolution limit (\ref{ray}) (see Remark \ref{rmk3}). 
This latter effect is called superresolution. 
 
Consider the noisy data model:
\beq
\label{u1}
Y=\bA X+E,\quad \|E\|_2\leq \ep
\eeq
where $X\in \IC^N$ is the object to be recovered, $Y\in \IC^m$
is the data vector and $E\in \IC^N$ represents  noise. 
We shall assume that $\bA$ has unit-norm columns. This can always be realized by redefining the object vector $X$. 

Sparse object reconstruction  for this model can be broken into two steps: localization (i.e. support recovery) and strength estimation. For underdetermined sytems, the former, being  combinatorial in nature, is by far
more difficult than the latter which is a straightforward inversion
if  the former is  exact.  
The former step is  called 
{\em model-selection }  in linear regression and machine learning 
theory \cite{BCJ, BBM, BM, BTW, CP,Gre,MB, Tib, ZY}
from which one of the reconstruction methods studied in the present paper originates.  

Exact localization with noisy data  is challenging. Many reconstruction methods guarantee 
stability (i.e. the reconstruction error bounded by a constant
multiple of  
the noise level) but not necessarily exact localization. 
Orthogonal Matching Pursuit (OMP) is a simple greedy algorithm with proven guarantee of exact localization for sufficiently
small noise and worst-case coherence.

A basic quantity for stability analysis in compressed sensing is the notion
of coherence.   Let the worst-case coherence $\mu(\bA)$ be defined as
\beq
\mu(\bA)=\max_{i\neq j} {\lt|\Phi^*_j\Phi_i\rt|\over \|\Phi_j\|_2\|\Phi_i\|_2}. 
\eeq

A standard  result is the following result  \cite{DET}.

\begin{proposition}
\label{prop0}
Consider the signal model (\ref{u1}). 
Suppose the sparsity $s$  of the real-valued object vector $X\in \IR^N$  satisfies
\[
s< {1\over 2} (1+{1\over \mu})-{\ep\over \mu X_{\rm min}},\quad
X_{\rm min}=\min_{i\in \cS} |X_i|
\]
Denote by $\hat X^\ep$ the output of OMP
which stops as soon as the residual error (in $\ell^2$-norm) is no greater than $\ep$. 
Then 
\begin{itemize}
\item[(i)] $\hat X^\ep$ has the correct support, i.e.
\beqn
\hbox{\rm supp} (\hat X^\ep)=\hbox{\rm supp} (X)
\eeqn
\item[(ii)]
$\hat X^\ep$ approximates the true object vector
\beqn
\|\hat X^\ep-X\|_2^2\leq {\ep^2\over 1-\mu(s-1)}
\eeqn
\end{itemize}
\end{proposition} 

The general lower bound
\cite{DGS, Wel}
\[
\sqrt{N-m\over m(N-1)}\leq \mu 
\]
for the mutual coherence of any $m\times N$ matrix $\bA$
implies that the sparsity $s$ allowed  by Proposition \ref{prop0}
is $\cO(\sqrt{m})$ for $N\gg m$. 

A main  purpose of the paper  is to explore the utility of
two other methods from compressed sensing theory,  
the One-Step Thresholding (OST) \cite{BCJ} and
the Lasso \cite{Tib, CDS},  that have the potential for exact localization of  much higher number $\cO(m)$ of objects.

The One-Step Thresholding (OST),    proposed in \cite{BCJ},   involves just  one matrix multiplication
plus {\em  thresholding}: Compute $Z=\bA^* Y$ and determine the set of
points
\[
\hat\cS=\lt\{i\in \{1,...,N\}: |Z_i|> \tau_*\rt\}
\]
for some threshold $\tau_*$. 
In other words, the OST is  the linear processor
of Matched Field Processing (MFP) plus
a thresholding step \cite{BKM}. On the other hand the linear
processor of 
MFP is the same as the first iterate of OMP. Consequently
OST has even lower complexity than OMP
which is its main appeal.

For the OST's performance guarantee,  
we need
the notion of average coherence defined as \cite{BCJ}
\beq
\nu(\bPhi)&=&{1\over N-1}
\max_{j'} \lt| \sum_{j\neq j'} \Phi_{j'}^*\Phi_j\rt|\nn
\eeq
in addition to the worst-case coherence. 


The following is the performance guarantee for OST \cite{BCJ}.
\begin{proposition}
\label{prop2} 

Consider the signal model (\ref{u1}). 
 Assume that $X\in \IR^N$ is drawn 
from the generic $s$-sparse ensemble of real-valued objects.
Assume $E$ to be distributed as $\hbox{\rm CN}(0,\sigma^2 \bI)$, the complex Gaussian random vectors with 
the covariance matrix $\sigma^2 \bI$.

Suppose
\beq
\mu(\bA) &\leq &{c_1\over \sqrt{m}} \leq {1\over \sqrt{10\log N}}\label{307}
%
\eeq
for some $c_1>0$ (which may depend on $\log N$)
and
\beq
 \nu(\bA) &\leq& {12\mu (\bA) \over \sqrt{m}}. \label{308}
\eeq
Assume $\|X\|_2=1$. 
Define the threshold 
\beq
\label{threshold}
\tau_*=4\sqrt{\log N}\max\lt\{\sigma, 12\mu \sqrt{2}\rt\}. 
\eeq
Suppose the number of objects obeying 
\beq
\label{300}
s\leq {m\over 2\log{N}}
\eeq
and that
\beq
\label{301}
X_{\rm min}=\min_{i\in \cS} |X_i|>2\tau_*. 
\eeq
Then 
the OST with threshold  $\tau_*$
satisfies $\IP\lt(\hat \cS \neq \cS\rt) \leq 9/N$.
\end{proposition}
In other words, for sufficiently small worst-case coherence
(\ref{307}) and  average coherence (\ref{308})
and noise (\ref{301}), OST can exactly  localize $\cO(m)$  objects, up to a logarithmic factor with high probability. 
Once
the support is exactly recovered, an estimate $\hat X$ can
be obtained by pseudo-inversion on the object support $\cS$.

The other method studied in this paper is the Lasso
\cite{Tib}.
The Lasso estimate $\hat X$ is defined as the solution
to
\beq
\label{lasso}
\min_{Z} {1\over 2} \|Y-\bA Z\|_2^2+\gamma\sigma \|Z\|_1,\quad\gamma>0
\eeq
where $\gamma$ is a regularization parameter. 

The following  sufficient condition for exact localization by the Lasso is given
by  \cite{CP}. 
\begin{proposition}
\label{prop1}

Consider the signal model (\ref{u1}). 
 Assume that $X\in \IR^N$ is drawn 
from the generic $s$-sparse ensemble of real-valued objects.
Assume $E$ to be distributed as $\hbox{\rm CN}(0,\sigma^2 \bI)$. 

Suppose that $\bPhi$  obeys the
coherence property 
\beq
\label{309}
\mu(\bPhi)\leq {a_0\over \log N}
\eeq
with some positive constant $a_0$. Suppose
\beq
\label{303}
s\leq{ c_0 N\over \|\bPhi\|_2^2 \log N}
\eeq
for some positive constant $c_0$. 
Let $\cS$ be the support of $X$ and suppose
\beq\label{305}
X_{\rm min} > 8\sigma\sqrt{2\log N}. 
\eeq
Then the Lasso estimate $\hat X$ with $\gamma=2\sqrt{2\log N}$
obeys
\beq
\hbox{\rm supp} (\hat X)&=&\hbox{\rm supp} (X)\\
\hbox{\rm sign}(\hat X)&=&\hbox{\rm sign}(X)
\eeq
with probability at least $1-2N^{-1}((2\pi \log N)^{-1/2}+sN^{-1})-\cO(N^{-2\log 2}))$. 
\end{proposition}

Some comparison between Proposition \ref{prop1} and \ref{prop2} is in order.
Both  deal with
randomly distributed  objects. 
Both (\ref{307}) and (\ref{309}) are sufficiently
weak assumptions for most imaging problems. Also (\ref{305}) and
(\ref{301}) are similar when $\mu=\cO(\sigma)$.
The lower bounds for the success probabilities are comparable up to
a logarithmic factor. The main technical assumption 
of Proposition \ref{prop1} is (\ref{303}) while
for Proposition \ref{prop2} it is  (\ref{308}). 
When  the operator norm $\|\bA\|_2$ obeys
the bound $\|\bA\|_2=\cO(N/m)$ condition (\ref{303}) 
is comparable to (\ref{300}).

A drawback to Proposition \ref{prop2} is that the thresholding rule (\ref{threshold}) requires
the precise knowledge of $\mu$ which can only
be calculated numerically. As we shall see,   the Lasso-based method also 
has a better numerical performance  than does  the OST (cf. Figure 
\ref{fig4} and \ref{fig6}). 

To realize the potential of the two above results
 in imaging, we shall consider the
 idea of  {\em random illumination}
 for point scatterers. We shall show  that a suitable condition of random illumination enables us to   (i) obtain a guaranteed exact localization 
of $s=\cO(m)$, up to a logarithmic factor, objects and to (ii) harness 
the
superresolution capability  (i.e. breaking the Rayleigh resolution limit (\ref{ray})). 

Previously we have studied  the problems of  imaging  point scatterers \cite{cis-simo, crs-pt} using coherence and operator-norm bounds. We shall demonstrate that the imaging performance
can be  significantly improved by random illumination. In particular, suitable 
random illumination leads to  superresolution. 

However,  both Propositions \ref{prop2} 
and \ref{prop1} share  the following common drawbacks: (i) they are restricted to random
objects; (ii) they do not address 
the reconstruction error when the error level is above threshold and exact localization is  unattainable; (iii) they are limited to the  i.i.d. Gaussian noise model. 
Issue  (i) is pertinent particularly to imaging extended objects
whose supports  are  clearly not random. Issue  (ii) 
is related to  robustness with respect to a wider range of error. 
Issue  (iii) arises  in optics where the Poisson or shot noise
model is more appropriate. 

The standard compressed sensing method that
are without any of the above limitations is the Basis Pursuit Denoising (BPDN)
\beq
\label{269}
\min_{Z} \|Z\|_1,\quad\hbox{s.t.}\,\, \|Y-\bA Z\|_2\leq \ep
\eeq
\cite{CDS}. BPDN, of course, is equivalent to the Lasso
(\ref{9}) for an appropriately chosen $\gamma$.

The performance guarantee for BPDN is typically given 
in terms of  the restricted isometry property (RIP) due to Cand\`es and Tao \cite{CT}.
Precisely, let  the sparsity $s$  of a vector  $Z\in \IC^N$ be the
number of nonzero components of $Z$ and define the restricted isometry constant (RIC)  $\delta_s\in [0,1]$
to be the {\em smallest} nonnegative  number such that the inequality
\beq
\label{rip}
 (1-\delta_s) \|Z\|_2^2\leq \|\bA Z\|_2^2\leq (1+\delta_s)
\|Z\|_2^2
\eeq
holds for all $Z\in \IC^N$ of sparsity at most $ s$. 
BPDN has the following performance guarantee  \cite{Can}.
\begin{proposition}\label{prop3}
Suppose the RIC satisfies the bound
\beq
\label{266}
\delta_{2s} < \sqrt{2}-1. 
\eeq
Then the BPDN minimizer $\hat X$ is unique and satisfies
the error bound
\[
\|\hat X -X\|_2 \leq C_1 s^{-1/2} \|X-X^{(s)}\|_1+C_2\ep
\]
where $X^{(s)}$ is the best $s$-sparse approximation of $X$ and $C_1, C_2$ are absolute constants depending on
$\delta_{2s}$ only. 
\end{proposition}
In Proposition \ref{prop3}, BPDN does not guarantee the exact recovery of the discrete support $X$, which is less important for
extended objects, but also does not
have any of the limitations mentioned above for Propositions \ref{prop2}
and \ref{prop1}.

The plan for the rest of the paper is as follows. 
In Section \ref{sec2}, we review the forward
scattering problem and the paraxial approximation. 
We describe the set-up of random illumination
in the paraxial regime. In Section \ref{main} we  state and prove  the main results. In Section \ref{sec:ext}  we  analyze the performance of BPDN with random
illumination for extended objects in  Section \ref{sec:ext}
and discuss the issue of resolution in imaging extended objects. In Section \ref{sec:thm3}
we give  the worst-case coherence bounds.
In Section \ref{sec5}, we give the average coherence bound. In Section \ref{sec:thm6} we give
an operator norm bound of $\cO(N/m)$ required
to guarantee a nearly optimal performance for
the Lasso. In Section \ref{sec:num} we present
numerical simulations to verify the predictions
and show the superiority of the Lasso over the OST
for the set-up of random illumination. We also present numerical results for extended objects. We conclude
in Section \ref{sec9}. 

\section{Point  scatterers and paraxial approximations}
\label{sec2}
Let $\cL$ be 
a finite square lattice of spacing $\ell$  in  the object  plane $\{z=0\}\subset\IR^3$:
\beq
\label{pix}
\cL=\lt\{\br_l: l=1,...,N\rt\}=\lt\{(i\ell , j\ell): i,j=1,...,\sqrt{N}\rt\},\quad l=(i-1)\sqrt{N}+j
\eeq
and suppose that $s$  point scatterers  are located
  at  grid points of $\cL$. 
The total number of grid points
 $N$ is  a perfect square. 

Let $\tau_{j}\in \IC, l=1,...,N$ be the reflectivity of the scatterers. 
The 
scattered  field  $u^{\rm s}$  obeys  
\beq
\label{311}
u^{\rm s}(\br) = \sum_{j=1}^N \tau_{j}
 G(\br,\br_{j})(u^{\rm i}(\br_{j})+u^{\rm s}(\br_j))
\eeq
for any $\br \not \in \{\br_k: \tau_k\neq 0\}$ 
where $u^{\rm i}$ is the incident field and
 \beq
\label{green}
G(\br,\br')={e^{\imath \om|\br-\br'|}\over 4\pi |\br-\br'|},\quad
\forall \br\neq \br'\in \IR^3
\eeq
is  the Green function
of the operator $-(\Delta+\om^2)$. 

In the Born scattering approximation, $u^{\rm s}$ on the
right hand side of (\ref{311}) is neglected,  resulting in
\beq
\label{312}
u^{\rm s}(\br) = \sum_{j=1}^N \tau_{j}
 G(\br,\br_{j})u^{\rm i}(\br_{j})
\eeq
\commentout{
The exciting field $u(\br_{j})$ is part
of the unknown and is governed by
the  Foldy-Lax equation
\beq
\label{FL}
u(\br_j)=u^{\rm i}(\br_j)
+\sum_{j\neq l} \tau_{l}
 G(\br_j,\br_{l})u(\br_{l}) 
\eeq
where the self-field term is omitted  to prevent blow-up. 
Under the Born approximation, we set $u(\br_j)=u^{\rm i}(\br_j), j=1,...,m$. 
}

Let $\ba_j, j=1,...,n$ be the locations of the sensors in the sensor plane $\{z=z_0\}\subset\IR^3$ and write $\ba_j=(\xi_j,\eta_j,z_0)$ where
$\xi_j$ and $\eta_j$ are chosen  independently and uniformly 
from the discrete subset of $[0,A]$
\beq
\label{arr}
\cD=\lt\{{q A \over \sqrt{N}}: q=1,...,\sqrt{N}\rt\}
\eeq
where $A$ is the aperture of the sensor array.

In the Fresnel approximation under  the condition 
\beq
\label{6-2}
{(A+\ell\sqrt{N})^4\over \lambda z^3_0}\ll 1
\eeq
the Green function $G$ can be approximated by  
\beq
\label{8-3}
\Gp(\br,\ba)={e^{\imath \om z_0}\over 4\pi z_0}  e^{\imath \om|x-\xi|^2/(2z_0)}e^{\imath \om|y-\eta|^2/(2z_0)},\quad
\br=( x,y, 0),\quad \ba=(\xi,\eta, z_0),  
\eeq
called  the paraxial Green function. 

In the subsequent analysis we shall assume both
the Born and paraxial approximations in the scattering
model.

   \begin{figure}[t]
\begin{center}
\includegraphics[width=0.8\textwidth]{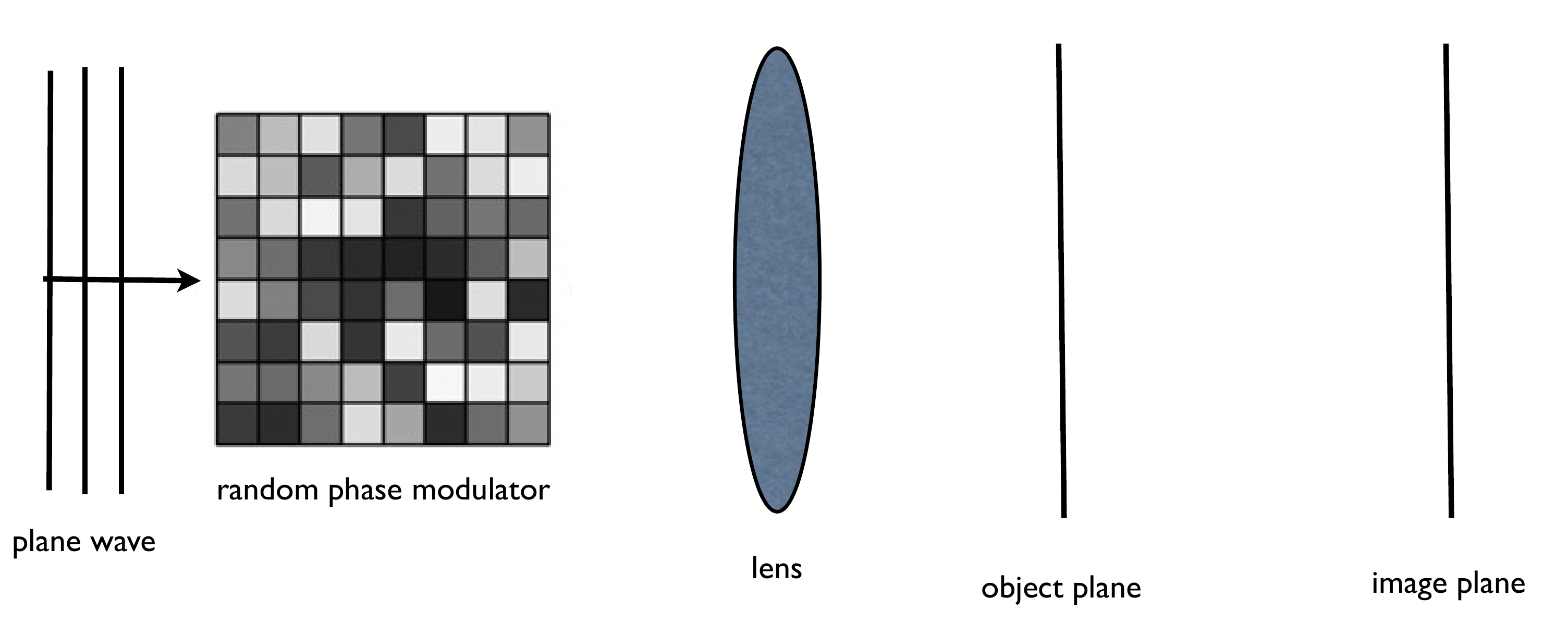}
\end{center}
\caption{The imaging geometry for point objects}
\label{geom}
\end{figure}

\commentout{
Assume we have a full control of
the source points in $\{(x,y,z): x, y\in [-L/2, L/2], z=z_1\}$ and write
the incident wave as
\beq
u^{\rm i}(\br)&=&\int_{-L/2}^{L/2} \int_{-L/2}^{L/2} \Gp(\br, (\xi,\eta, z_1)) f(\xi,\eta) d\xi  d\eta\nn\\
&=&{e^{\imath \om z_1}\over 4\pi z_1}\int_{-L/2}^{L/2} \int_{-L/2}^{L/2}    e^{\imath \om|x-\xi|^2/(2z_1)}e^{\imath \om|y-\eta|^2/(2z_1)}f(\xi,\eta)d\xi d\eta.\label{19}
\eeq

Let the source distribution $f$ 
be  a complex-valued, circularly symmetric Gaussian white-noise field of variance $\kappa^2$:
\beq
\IE \lt[f(\xi,\eta)f^*(\xi',\eta')\rt]&=&\kappa^2 \delta(\xi-\xi', \eta-\eta')\\
\IE \lt[f(\xi,\eta)f(\xi',\eta')\rt]&=&0,\quad \forall \xi,\xi',\eta,\eta'.
\eeq

Note that the transformation defined in (\ref{19}) is unitary
and hence $u^{\rm i}$ is also a complex-valued, circularly
symmetric Gaussian random field. Let us 
calculate its correlation  as follows.
\beq
{\IE[u^{\rm i}(x,y,0) u^{\rm i}(x', y', 0)]}=0,\quad
\forall x,x',y,y', 
\eeq
and 
\beq
\lefteqn{\IE[u^{\rm i}(x,y,0) u^{\rm i*}(x', y', 0)]}\label{corr}\\
&=& {\kappa^2\over 16 \pi^2 z_1^2}
 \int_{-L/2}^{L/2} \int_{-L/2}^{L/2}    e^{\imath \om|x-\xi|^2/(2z_1)}e^{\imath \om|y-\eta|^2/(2z_1)}  e^{-\imath \om|x'-\xi|^2/(2z_1)}\nn\\
 &&\hspace{5cm} \times e^{-\imath \om|y'-\eta|^2/(2z_1)}d\xi d\eta\nn\\
&=& {\kappa^2\over 16 \pi^2 z_1^2}e^{\imath \om (x^2-x^{'2})/(2z_1)} e^{\imath \om (y^2-y^{'2})/(2z_1)} \nn\\
&&\hspace{3cm}
\times \int_{-L/2}^{L/2} \int_{-L/2}^{L/2}    e^{\imath \om (x-x')\xi/z_1}    e^{\imath \om (y-y')\xi/z_1}d\xi d\eta\nn\\
&=& {\kappa^2\over 4 \pi^2\om^2(x-x')(y-y')}e^{\imath \om (x^2-x^{'2})/(2z_1)} e^{\imath \om (y^2-y^{'2})/(2z_1)}\nn\\
&&\hspace{4cm}
\times \sin{\lt(\om(x-x')L\over 2 z_1\rt)}\sin{\lt(\om(y-y')L\over 2z_1\rt)}.  \nn
 \eeq
 For $(x,y), (x',y')\in \cL$, $(x-x')/\ell, (y-y')/\ell \in \IZ$.
 Hence (\ref{corr}) vanishes  for $(x,y)\neq (x', y')$ 
 under the assumption 
 \beq
 \label{18}
 {\ell L\over \lambda z_1}\in \IN,\quad\lamb=2\pi/\om.
 \eeq
 Therefore under (\ref{18}) the random incident field
 takes on i.i.d. random values at grid points. 
}

 
 A main ingredient of the proposed approach is
{ random } illumination which  has recently been used extensively
for wavefront reconstruction and imaging \cite{Alm, Rice, Rom}.  Here we consider random {\em phase} modulation (RPM) which is a random perturbation of the phase of a wavefront while maintaining the amplitude of the near field beam almost constant. The advantage of phase modulation, compared to amplitude modulation, is the lossless energy transmission of an incident wavefront through the modulator. 
In optics RPM can be created by random phase plates, digital holograms 
or liquid crystal panels \cite{BWW, SW}.

 \section{Main results for point objects}\label{main}
 We assume that as a result of $p$ independent realizations of 
 random phase modulators
 the incident field at the grid points can be represented as $e^{\imath \theta_{kj}}, k=1,...,p, j=1,...,N$ where  $\theta_{kj}$ are i.i.d
 {\em uniform}  random variables in $[0,2\pi]$ (i.e. circularly symmetric). 
The information about  $ \theta_{kj}$ is incorporated in
the sensing matrix. 

 Let the scattered field $u^{\rm s}_k$ is measured and collected by $n$ sensors located at $\ba_l, l=1,...,n$. 
Let $X=(\tau_j)_1^N\in \IC^N$ be the object vector  and  $Y=(Y_i)=(u_k^{\rm s}(\ba_j)) \in \IC^{np}, i=(k-1)n+j, j=1,...,n, $ the data vector.

 After proper normalization, the data vector $Y$ can be written as (\ref{u1}) 
with the sensing matrix  $\bPhi$ being  the column-normalized
version of $[\Gp(\ba_l, \br_j)u^{\rm i}_k(\br_j)]$, i.e. 
 \beq
 \label{1.10}
 \phi_{i j} &=& {1\over \sqrt{np}} e^{\imath \om|x_j-\xi_l|^2/(2z_0)}e^{\imath \om|y_j-\eta_l|^2/(2z_0)}e^{\imath \theta_{kj}},\quad  i=(k-1)n+l.
 \eeq
 Here $m=np$ is the number of data.
   
 \commentout{ where  the error term $E$
includes the background noise and 
 the mismatch between the
exact Green function 
(\ref{green}) and the paraxial Green function (\ref{8-3}). 
Note that the latter part $E^p=(E^p_i)$ of the errors 
\[
E^p_i={4\pi z_0 e^{-\imath \om z_0}
\over \sqrt{n}} \sum_{j=1}^N \lt(G(\ba_i,\br_j)-\Gp(\ba_i,\br_j)\rt) u^{\rm i}(\br_j) X_j, \quad \forall i=1,...,n,
\]
depend explicitly on the unknown vector $X$. 
}

 Our first result is  a performance guarantee
 for the OST with random illumination in the diffraction-limited
 case satisfying the Rayleigh resolution criterion.
\begin{theorem}\label{thm1} Let 
\beq
\label{mesh}
N^2\leq {\delta\over 2} e^{K^2/2},\quad \delta, K>0. 
\eeq
Suppose
\beq
\label{27}
{np} \geq 40 K^4 \log{N}
\eeq
and 
\beq
\label{apert}
{A \ell \over \lambda  z_0}=1. 
\eeq
Then with probability at least  
\beq
\label{136}
1-2\delta-4t\sqrt{2\over \pi}-{4\over \sqrt{p}}-{4\over \sqrt{n}}-8N e^{-12 t^2\sqrt{N-1\over np}},\quad\forall t>0
\eeq
 OST with
the threshold (\ref{threshold})
can localize exactly $s$ objects satisfying (\ref{300})-(\ref{301}). 
\end{theorem}

\begin{remark}
\label{rmk6}
The constants $\delta$ and $K$ in (\ref{mesh}) are controlling parameters. $\delta$ can be  adjusted to control
the lower bound (\ref{136})  for success probability
and then $K$ can be adjusted to control the 
number of grid points in the computation domain
and the number of data. 

For example, suppose $\delta=1\%$ is acceptable. 
Then (\ref{mesh}) with $K=10$ implies a computation domain
of about $0.1 e^{25}/\sqrt{2}$ grid points.

\end{remark}
\commentout{
\begin{remark}\label{rmk1}
For the lower bound (\ref{136})  to approach unity, it is necessary to 
let $N, p, n\to\infty$ (and $t,\delta\to 0$ accordingly)  in a suitable manner (e.g., $ N\gg np$). 

\end{remark}
}

\begin{proof}
The proof of Theorem \ref{thm1} relies on Proposition \ref{prop2}
and 
the following three lemmas.

\begin{lemma}\label{lem1}
Under  (\ref{mesh}),  the worst case coherence satisfies
\beq
\label{321}
\IP\lt\{\mu(\bPhi)\leq {a K\sqrt{2}\over \sqrt{p} } + {2 K^2\over \sqrt{np}}\rt\}\geq 1-2\delta 
\eeq
where $a$ is given by (\ref{135}).

In particular,  if
(\ref{apert}) holds 
then $a=0$ and (\ref{321}) becomes 
\beq
\label{30}
\IP\lt\{\mu(\bPhi)\leq {2} K^2/\sqrt{np}\rt\}\geq 1-2\delta. 
\eeq
\end{lemma}

The proof of Lemma \ref{lem1} is given in Section \ref{sec:thm3}. 
The utility of estimate (\ref{321})  lies in the situation where both
the aperture and the sensor number are limited but
the number of probe waves is exceedingly large (see Remark \ref{rmk3}).  For the proof of Theorem \ref{thm1} we need the estimate (\ref{30}).

\begin{lemma}
\label{lem:lb}Under the assumption (\ref{apert}), 
\beq
\IP\lt[\mu(\bA) \geq  {2t_1t_2\over\sqrt{{np}}} \rt] \geq \lt(1-2t_1\sqrt{2\over \pi} -{4\over \sqrt{p}}\rt)\lt(1-2t_2\sqrt{2\over \pi} -{4\over \sqrt{n}}\rt). \label{32}
\eeq
\end{lemma}
 Lemma \ref{lem:lb} is an easy consequence of  the Berry-Esseen theorem and its proof is given in Section \ref{sec5.2}.  

\begin{lemma}\label{lem2} Let (\ref{apert}) hold true. 
Then for any $c>0$
\beq
\label{32'}
\IP\lt\{\nu(\bA)\leq 
{c\over {np}} \rt\}\geq 1-8Ne^{-{c\over 2}\sqrt{{N-1\over np}}}.
\eeq
\end{lemma}
The proof of Lemma \ref{lem2} is given in Section \ref{sec5}. 

First of all, by the upper bound (\ref{321}) for the worst case coherence  and setting  $c_1=2K^2$ in (\ref{307})  then the first
inequality of
(\ref{307}) holds with probability at least $1-2\delta$.
The second inequality of (\ref{307}) follows from  (\ref{mesh}) and (\ref{27})  and holds with probability at least $1-2\delta$. 

Second, the lower bound (\ref{32}) for the worst case coherence, with $t_1=t_2=t$,  
and the upper bound (\ref{32'}),  with  $c=24t^2$,  for the average coherence 
imply that (\ref{308}) holds with probability at least
\[
1-4t\sqrt{2\over \pi}-{4\over \sqrt{p}}-{4\over \sqrt{n}}-8N e^{-12 t^2\sqrt{N-1\over np}}.
\]
This completes the proof of Theorem \ref{thm1}.
\end{proof}

Our second result is a performance guarantee for the
Lasso with random illumination.  
\begin{theorem}Let (\ref{mesh}) hold and \label{thm2}
suppose
\beq\label{139}
{aK\sqrt{2}\over \sqrt{p}}+{2K^2\over \sqrt{np}}\leq {a_0\over \log{N}}
\eeq
where
\beq
\label{135}
a=\max_{j\neq j'}\lt|\IE \lt(e^{\imath \xi_l\om(x_{j'}-x_j)/z_0}\rt)
\IE\lt(e^{\imath \eta_l\om(y_{j'}-y_j)/z_0}\rt)\rt|.  
\eeq
Assume that the $s$ objects are real-valued and 
satisfy (\ref{305}) and 
\beq
\label{138}  s\leq {c_0np\over 2\log{N}}. 
\eeq
Then the Lasso estimate $\hat X$ with $\gamma=2\sqrt{2\log{N}}$ has the same support as $X$ with probability 
at least 
\beq
\label{137}
&&1-2\delta- \rho n(n-1){\pi \over 2}\sqrt{np-1\over N}-2n^2p(p-1) e^{-{N\over (np-1)^2}}\\
\nn&& \quad -2N^{-1}((2\pi \log N)^{-1/2}+sN^{-1})-\cO(N^{-2\log 2})). 
\eeq
 
\end{theorem}
\begin{remark}
\label{rmk1}

While it
requires  that $N\gg np$
for the bound (\ref{136}) to approach unity, 
it demands  a much stronger assumption  $N\gg \max\{pn^5, p^2n^2\}$
for the bound
(\ref{137}) to behave the same way. 
Numerical evidences
indicate the latter to be a pessimistic estimate. 

For the special case of single  sensor $n=1$, the  probability lower bound
(\ref{137}) is substantially improved and
requires $N\gg p^2$ to approach unity.
On the other hand, for (\ref{136}) to approach unity,
it is necessary that $n\to \infty$ (hence $n=1$ is not an option). 

\end{remark}
\begin{remark}
\label{rmk3}
The superresolution effect can occur when the number $p$  of random probes is large. Consider, for example, the case of  $n=1$ and hence
the aperture $A$ is essentially zero. Since $a\leq 1$, the condition
\[
{K\sqrt{2}+2K^2\over \sqrt{p}} \leq {a_0\over \log{N}}
\]
and
\[
s\leq {c_0p\over 2\log{N}}
\]
implies that the Lasso  with $\gamma=2\sqrt{2\log{N}}$ 
recovers exactly the support of $s$ objects with
probability at least that given by (\ref{137}). 

This superresolution effect should be compared to
that with  deterministic near-field illumination \cite{subwave}. 
\end{remark}

\begin{proof}
The proof of Theorem \ref{thm2} uses
Proposition \ref{prop1}, 
Lemma \ref{lem1} and the following operator-norm bound.
 
\begin{lemma}\label{lem3}
We have 
\beq
\label{202}
\IP\lt\{\|\bA\|^2_2< {2N\over np}\rt\}\geq 1-
\rho n(n-1){\pi \sqrt{np-1}\over 2\sqrt{N}}-2n^2p(p-1) e^{-{N\over (np-1)^2}}. 
\eeq

\end{lemma}

On one hand,  Lemma \ref{lem1} and (\ref{139}) imply that
(\ref{309}) holds with probability at least
$1-2\delta$. 

On the other hand, Lemma \ref{lem3} and  (\ref{138}) imply that
(\ref{303}) holds with probability at least given by the right hand side of  (\ref{202}).

Combining the two and using Proposition \ref{prop1} we obtain 
the desired statement of Theorem \ref{thm2}. 

\end{proof}
To further demonstrate the advantage of random illumination,
let us consider the imaging  set-up of {\em multistatic responses} (MR) which  consists of an array of  $n$ fixed   transceivers which
 are both sources and sensors (i.e. transceivers).   One by one, each transceiver of the array emits  an impulse  and the entire array of transceivers 
records the echo. Each transmitter-receiver pair gives rise to
a datum and there are altogether $n^2$ data forming a
data matrix called the multistatic response matrix. 
By the reciprocity of the wave equation, the MR matrix
is symmetric and hence has at most $n(n+1)/2$ degrees of freedom. 

Recalling  the coherence and operator norm bounds established
in \cite{crs-pt} and using Proposition \ref{prop1} as in the proof
of Theorem \ref{thm2} (below), we have
the following result \cite{crs-pt} analogous to Theorem \ref{thm2}.

\begin{proposition}
\label{prop4}
Let  the locations of the $n$ transceivers be  i.i.d. uniform
random variables in $[0,A]^2$. Let (\ref{27}) and (\ref{apert}) hold true. 

Suppose
\[
n\geq {K^2\log{N} \over a_0}
\]
and that the $s$ real-valued objects satisfy (\ref{305}) and
\beq
\label{162}
s\leq {c_0n(n+1)\over 4\log{N}}.
\eeq
Then the Lasso estimate $\hat X$ with $\gamma=2\sqrt{2\log{N}}$ has the same support as $X$  with probability at least 
\beq
\label{161}
1-2\sqrt{2\delta}-{\rho n^{5/2}(n+1)^{5/2}\over \pi 2^{5/2} N^{1/2}} -2N^{-1}((2\pi \log N)^{-1/2}+sN^{-1})-\cO(N^{-2\log 2})).
\eeq
\end{proposition}
\begin{remark}
\label{rmk2} The main drawback of the lower bound (\ref{161}) lies in 
the third term which requires $N\gg n^{10}$ to diminish.

More generally, one can consider the case of $p$ transmitters
and $n$ receivers, all randomly and independently distributed 
in $[0,A]^2$. Then an extension of  the bound (\ref{161}), which is omitted here, requires
$N\gg n^5 p^5$ (cf. \cite{cis-simo}). 

\end{remark}

From dimension count, a fair comparison with Theorem \ref{thm2} would
be to set $p=(n+1)/2$ and match their degrees of freedom,
i.e. $n(n+1)/2$.   However, Proposition \ref{prop4} does
not guarantee superresolution  when (\ref{apert}) is violated
preventing the worst case coherence from being sufficiently small due to the deterministic nature of the illumination.
Also, the probability lower bound (\ref{161}) has
a less favorable scaling behavior ($N\gg n^{10}$)  than (\ref{137}) for  $p=(n+1)/2$ ($N\gg n^6$, cf. Remark \ref{rmk1}). Indeed, the numerical simulations
show the recovery with random illumination has
a higher success rate than the MR recovery
(Figures \ref{fig1} and \ref{fig1'}). 

\commentout{
\begin{proposition}\cite{BCJ}
\label{prop:bcj}
Suppose 
\beq
\label{cp1}
\mu(\bPhi)&=&{C_1\over \sqrt{np}}\leq {1\over \sqrt{10\log{m}}}\\
\nu(\bPhi)&\leq &{12 \mu(\bPhi)\over \sqrt{np}}\label{cp2}
\eeq
where $C_1>0$ may depend on $\log m$. Assume $\|X\|_2=1$. 
Then the time reversal scheme with the threshold
\[
\rho=4\max{\{12 \mu \sqrt{2\log m}, \sigma\sqrt{\log m}\}} 
\]
satisfies 
\[
\IP\lt[ \hat S\neq \supp (X)\rt]\leq {9\over N}
\]
if
\[
np> s\cdot \max\lt[2\log m, {64\log m\over 
\hbox{\rm SNR}_{\rm min}}, {18432 C_1^2 \log m\over \hbox{\rm MAR}}\rt]
\]
where
\beq
\hbox{\rm MAR}&=&s\cdot\min_{X_i\neq 0} |X_i|^2\\
\hbox{\rm SNR}_{\rm min}&=& {\hbox{\rm MAR}\over 
\sigma^2}. 
\eeq
\end{proposition}
Since $\|X\|_2=1$, $\hbox{\rm MAR}$ has the meaning of
minimum-to-average ratio. 
}

  \section{Sparse extended objects}\label{sec:ext}
  
     \begin{figure}[t]
\begin{center}
\includegraphics[width=0.8\textwidth]{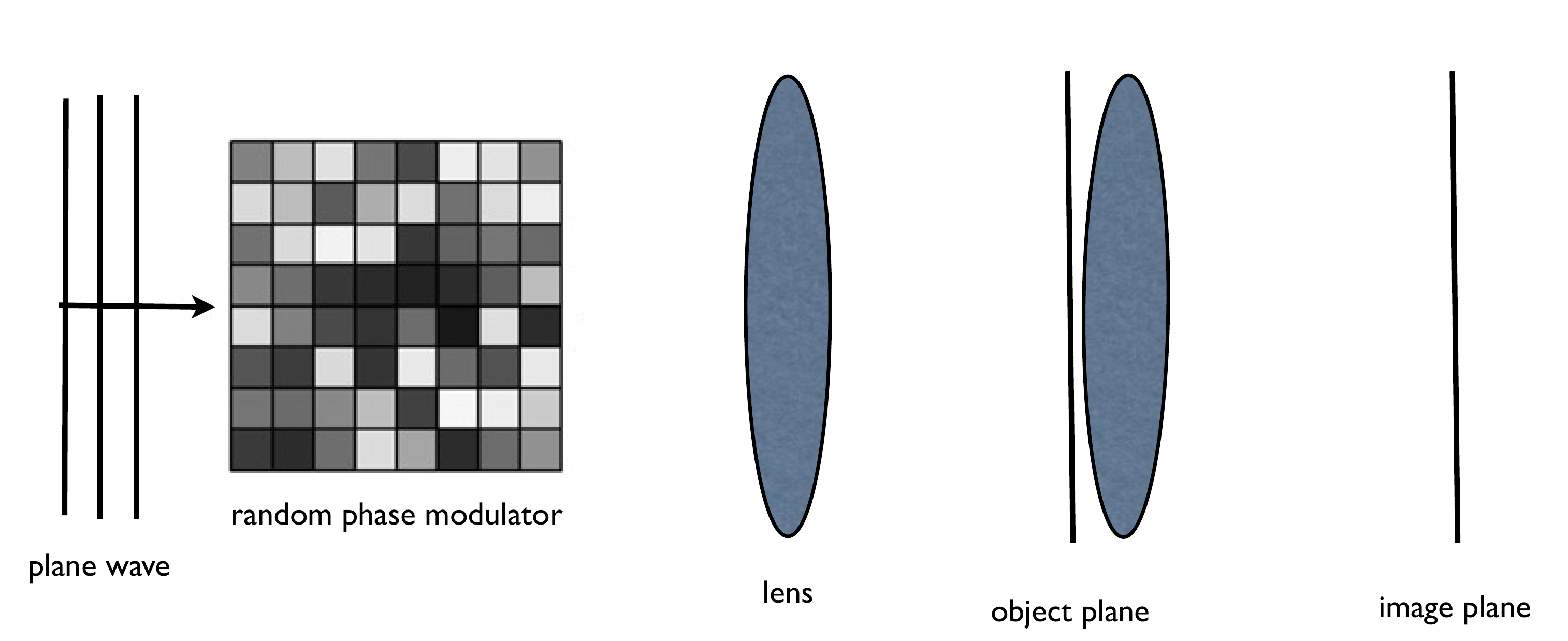}
\end{center}
\caption{The imaging geometry for extended objects}
\label{geom2}
\end{figure}

We extend the above results to the case of sparse extended objects here (Figure \ref{geom2}).  

We pixelate  the sparse extended  object with $N$ pixels $\square_j, j=1,...,N$  of size $\ell$ 
to create a  piecewise
constant approximation of the object. The centers of the  pixels
are identified as $\cL$ given in (\ref{pix}). Let $O(\br)$ be a the original  object and $O_\ell$ its $\ell$-discretization, i.e.
\[
O_\ell=\sum_{j=1}^N \II_{\square_j} O(\br_j)
\]
where $\II_{\square_j}$ is the indicator function of the pixel $\square_j$. 
We
reconstruct the discrete approximation $O_\ell$ by determining  the object function restricted to $\cL$, denoted
still by $X=(O(\br_j))$, by compressed sensing techniques.

Under the random illumination $u^{\rm i}_k$, pixel $\square_j$ now produces
a signal at the sensor $\ba_l$ of the form
\beqn
O(\br_j) \int_{\square_j} \Gp(\br, \ba_l) e^{-\imath \om x^2/(2z_0)}e^{-\imath \om y^2/(2z_0)}  u^{\rm i}_k(x,y) dx dy
\eeqn
where the quadratic phase factors are due to
the presence of a parabolic lens immediately after
the object plane (Figure \ref{geom2}). This lens is introduced here to simplify our analysis. In practice, 
the lens is not needed and should have
a negligible effect on performance.

As for the case of point objects we   assume
that as a result of the RPM $u^{\rm i}_k$ takes a constant value  $e^{\imath \theta_{kj}}$
in pixel $\square_j$ and that $\theta_{kj}$ are i.i.d. random
variables in $[0, 2\pi]$ as a result of random phase modulation.  

The total signal produced by $O_\ell$ and detected at sensor $\ba_l$ is
\beqn
&&\sum_{j} O(\br_j) e^{\imath \theta_{kj}}\int_{\square_j} \Gp(\br, \ba_l) e^{-\imath \om x^2/(2z_0)}e^{-\imath \om y^2/(2z_0)} dx dy\\
&=&\sum_{j} O(\br_j)  e^{\imath \theta_{kj}} e^{\imath \om \xi_l^2/(2z_0)}e^{\imath \om\eta_l^2/(2z_0)} e^{-\imath \om \xi_lx_j/z_0}e^{-\imath \om \eta_ly_j/z_0}   
\int_{\square}  e^{-\imath \om \xi_lx/z_0}e^{-\imath \om \eta_ly/z_0}  dx dy,
\eeqn
plus an error term $E_{kl}$ which includes the discretization error and
external noise where $\square$ denotes the square
of size $\ell$ centered at the origin. Since
\beq
\label{260}
\int_{\square}  e^{-\imath \om \xi_lx/z_0}e^{-\imath \om \eta_ly/z_0}  dx dy=
{2z_0\over \om \xi_l \ell} \sin{\lt(\om\xi_l\ell \over 2 z_0\rt)}
{2z_0\over \om \eta_l \ell} \sin{\lt(\om\eta_l\ell \over 2 z_0\rt)}\equiv g(\ba_l)
\eeq
independent of the pixel index, we can normalize the data by dividing  the
signal at sensor $l$ by this number as long as
\beq
\label{261}
{\xi_l\ell \over \lambda z_0},\,\,{\eta_l\ell \over \lambda z_0}< 1,\quad
\forall l=1,\ldots, n.
\eeq

Dividing  the data further by the phase factors
$ e^{\imath \om \xi_l^2/(2z_0)}e^{\imath \om\eta_l^2/(2z_0)}$ and $\sqrt{np}$,  
we write the signal model as (\ref{u1}) with the sensing matrix element 
\beq
\label{ext1}
\phi_{ij}
&=& {1\over \sqrt{np}} e^{\imath \theta_{kj}} e^{-\imath \om \xi_lx_j/z_0}e^{-\imath \om \eta_ly_j/z_0},\quad
i=(k-1) n+l. 
\eeq

The difference between the signals produced by $O$ and its discretization $O_\ell$ is the discretization error $E_{\rm disc}$. 
How small must $\ell$ be in order
for the $\ell_2$-norm of the discretization error  $E_{\rm disc}$ be
less than, say, $\ep$ after rewriting the signal model as (\ref{u1})? 
This can be estimated as follows. 

First, by the inequality $\|E_{\rm disc}\|_2\leq \|E_{\rm disc}\|_\infty \sqrt{np}$
it suffices to show $\|E_{\rm disc}\|_\infty \leq \ep/\sqrt{np} $.

Since  
\[
u^{\rm i}_k(\br)=\sum_{j=1}^N e^{\imath \theta_{kj}} \II_{\square_j}(\br)
\]
 is the illumination field, the uncontaminated signal detected by sensor $\ba_l$ in the absence of external noise in the signal model
 (\ref{u1})  is
\beq
\label{271}
 (\cF O)_i ={1\over  g(\ba_l)} \sum_{j=1}^N e^{\imath \theta_{kj}}   \int_{\square_j} O(x,y) e^{-\imath \om \xi_lx_j/z_0}e^{-\imath \om \eta_ly_j/z_0}  dx dy, 
 \eeq
 for $i=(k-1)n+l.$
 On the other hand we have
 \[
 (\cF O_\ell)_i ={1\over  g(\ba_l)} \sum_{j=1}^N e^{\imath \theta_{kj}}  O(\xi_j,y_j)  \int_{\square_j}e^{-\imath \om \xi_lx_j/z_0}e^{-\imath \om \eta_ly_j/z_0} dx dy
 \]
 for $ i=(k-1)n+l. $
By definition 
  \[
 E_{\rm disc}= \cF O -\cF O_{\ell}\in \IC^{pn}
 \]
 and hence 
 \beq
 \label{19}
 \|E_{\rm disc}\|_\infty&\leq& {\|O-O_{\ell}\|_{L^1}\over  \min_{l}| g(\ba_l)|}
\eeq
where $\|\cdot\|_{L^1}$ denotes
\[
\|f\|_{L^1}=\int |f(x,y)| dx dy, 
\]
i.e.  the norm of  the  function space $L^1$. 
Therefore we have the following statement.
\begin{lemma}
\label{lemm5}
If  
 \beq
 \label{270}
 {\|O-O_{\ell}\|_{L^1}}\leq {\ep\over \sqrt{np}}  \min_l |g(\ba_l)|
 \eeq
 then
\[
\|E_{\rm disc}\|_2 \leq \ep.
\]
\end{lemma} 
 \begin{remark}
 The presence of the factor $(np)^{-1/2}$ in (\ref{270}) is due to
 the transition from  $L^1$ function space norm to the discrete $\ell_2$-norm. 
 \end{remark}

Since the sensing matrix  (\ref{1.10}) for the point objects can be written as 
\[
 \bD_1 \bA \bD_2
 \]
 where $\bA$ is as (\ref{ext1}) and 
 \beqn 
 \bD_1&=&\hbox{diag} (e^{\imath \om \xi_l^2/(2z_0)} e^{\imath \om \eta_l^2/(2z_0)})\\
 \bD_2&=&\hbox{diag} (e^{\imath \om x_j^2/(2z_0)} e^{\imath \om y_j^2/(2z_0)})
 \eeqn
 are diagonal, unitary matrices. All the preceding results, including 
 Theorems \ref{thm1}
and \ref{thm2},  can be proved for the sensing matrix (\ref{ext1}) 
by minor modification of the previous arguments.

However, the object vector $X=(O(\br_j))$ of an extended object generally does not fall into the category of {\em random point} objects assumed
in either Proposition \ref{prop2} or \ref{prop1} since
by definition the discrete approximation of an extended object must cluster in aggregates and its amplitude typically changes continuously.  So we take an alternative approach below
by resorting to 
the minimization  principle (\ref{269}) of BPDN.

The RIC for a structured sensing matrix such as (\ref{ext1}) is
difficult to estimate directly except for the case of
single shot ($p=1$) and the case of one sensor ($n=1$). 
For the one-sensor case, (\ref{ext1}) with $(\xi_l,\eta_l)=(0,0)$ is the complex-value version 
of the random i.i.d. Bernoulii matrix:
\beq
\phi_{kj}={1\over \sqrt{p}} e^{\imath\theta_{kj}},
\eeq
whose RIC can be 
easily estimated by the same argument given in \cite{BDD}. 
The single sensor imaging set-up resembles that of
Rice's single-pixel camera \cite{Rice} which employs 
a discrete random screen instead of a random phase
modulator.    

For the single-shot case, the sensing matrix (\ref{ext1}) is equivalent to the random partial Fourier matrix, modulo
an unitary diagonal matrix, and the standard RIP estimate \cite{Rau}
requires the Rayleigh criterion (\ref{apert}) to be met which guarantees (\ref{261}) with probability one. However,
there exists  a small probability of $\ba_l=(\xi_l,\eta_l)$ falling near the boundary of the aperture and hence a small value of $|g(\ba_l)|$. Normalizing the data by $|g(\ba_l)|$  then carries 
a small risk of magnifying the errors.

For the general set-up with  multiple shots and sensors,  we  use the
mutual coherence to bound the RIC trivially as follows.

\begin{proposition}\label{prop40} For any $s\in \IN$ we have
\[
\delta_s\leq \mu(\bA) (s-1).
\]
\end{proposition}

Combining Lemma \ref{lem1},  Propositions  \ref{prop3} and  \ref{prop40} we obtain the following result.
\begin{theorem}
\label{thm3}

 Under  (\ref{mesh}),   the RIC bound (\ref{266}) holds true 
with probability at least $1-2\delta$ 
for the sensing matrix (\ref{ext1}) and sparsity up to
\beq
\label{321'}
s< {1\over 2} +\lt({1\over \sqrt{2}}-{1\over 2}\rt)\lt(  {a K\sqrt{2}\over \sqrt{p} } + {2 K^2\over \sqrt{np}}\rt)^{-1}
\eeq
where 
\[
a=\max_{j\neq j'}\lt|\IE \lt(e^{\imath \xi_l\om(x_{j'}-x_j)/z_0}\rt)
\IE\lt(e^{\imath \eta_l\om(y_{j'}-y_j)/z_0}\rt)\rt| 
\]
c.f. (\ref{135}). 

Furthermore, suppose the total error in the data is $E=E_{\rm disc}+E_{\rm ext}$ where
 $E_{\rm disc}$ and  $E_{\rm ext}$ are, respectively, the discretization error and 
the external noise.  Then the reconstruction  $\hat X$ by BPDN satisfies the error bound
\beq
\label{280}
\|\hat X -X\|_2 \leq C_1 s^{-1/2} \|X-X^{(s)}\|_1+C_2\lt( \|E_{\rm disc}\|_2+\|E_{\rm ext}\|_2\rt)
\eeq
for all $s$ satisfying (\ref{321'}). 

\end{theorem}
\begin{remark}
Since BPDN does not guarantee exact localization,  an appropriate
metric for resolution can be formulated in terms of  the smallest pixel size $\ell_{\rm min}$ and largest sparsity $s$ such that (\ref{280}) holds true with both the discretization error $E_{\rm disc}$ and $s^{-1/2}\|X-X^{(s)}\|_1$ being reasonably small. 

The right definition of ``small errors", however, is problem specific.
The discrete norms ($\ell_1$- or $\ell_2$- norm)
tend to go up simply because the effective sparsity  increases. 
Hence the right metric of reconstruction  error should be properly
normalized by the size of the object. For example, consider the
special case when $X$ is
$s$-sparse. Then we can rewrite (\ref{280}) as 
\beq
\label{280'}
s^{-1/2} \|\hat X -X\|_2 \leq C_2\|E_{\rm ext}\|_2s^{-1/2}+C_2\|E_{\rm disc}\|_2 s^{-1/2}
\eeq
whose left hand side is a measure of  the reconstruction  error {\em per 
pixel of size $\ell$}.  

Below the diffraction limit $A\ell/(\lambda z_0)<1$ ($a\neq 0$), 
one can reduce the discretization error by reducing the pixel size
according to Lemma \ref{lemm5}. On the other hand, the sparsity $s$ increases in proportion to $\ell^{-2}$ for a two-dimensional extended object. To satisfy (\ref{321'}) the smallest admissible  pixel size $\ell_{\rm min}$ is bounded from below roughly by 
\beq
\label{280''}
\ell_{\rm min} \stackrel{>}{ \sim} {a^{1/2} p^{-1/4}}
\eeq
meaning that  the minimum super-resolved scale decreases at least as fast
as the negative quarter power of the number of random illuminations. 

For the diffraction-limited case $a=0$, we have instead
\[
\ell_{\rm min} \stackrel{>}{ \sim} {n^{-1/4} p^{-1/4}}
\]
which is more favorable  than (\ref{280''}) for $n\gg 1$. However,
 the discretization error  bound (Lemma \ref{lemm5}) is less
useful in this case. 
\end{remark}
\commentout{
\begin{remark}
For the estimate  (\ref{270}) to be useful, the ratio $A\ell/(z_0\lambda)$ must be less than 1, rendering $a
\neq 0.$ In this case, for $n, p\gg 1$ the dominant term
on the right hand side of (\ref{321'}) is
\[
\lt({1\over {2}}-{1\over 2\sqrt{2}}\rt){\sqrt{p}\over aK},
\]
i.e. the sparsity $s=\cO(\sqrt{p})$. 

\end{remark}
}

\section{Worst-case  coherence bound}
\label{sec:thm3}

\subsection{Proof of Lemma \ref{lem1}: upper bound}
\begin{proof}
Summing over $\ba_l, l=1,...,n$ we obtain 
\beq
\label{1.20}
\sum_{k=1}^p\sum_{l=1}^n \phi^*_{i{j'}}\phi_{ij}
&=&e^{\imath \om(x_{j}^2+y_{j}^2-x_{j'}^2-y_{j'}^2)/(2z_0)}{1\over np}  \sum_{k=1}^pe^{\imath (\theta_{kj}-\theta_{kj'})}\sum_{l=1}^n e^{\imath \xi_l\om(x_{j'}-x_j)/z_0}
e^{\imath \eta_l\om(y_{j'}-y_j)/z_0}.
\eeq
We shall estimate the two summations separately. 

First consider 
the  summation over random illuminations $k=1,...,p$. 
Define the random variables $A_l, B_l, l=1,...,n$, as 
\beq
\label{9.1'}
A_l&=&\cos{\lt[\theta_{kj}-\theta_{kj'}\rt]}\\
B_l&=&\sin{\lt[\theta_{kj}-\theta_{kj'}\rt]}\label{9.2'}
\eeq
and
let
\beq
\label{9}
S_p=\sum_{l=1}^p (A_l+iB_l).
\eeq

To estimate $S_p$, we recall the Hoeffding inequality
\cite{Hoe}.

\begin{proposition}
Let $A_1+iB_1,  ..., A_p+iB_p$ be independent random variables. Assume
that $A_l, B_l \in [a_l, b_l], l=1,...,p$ almost surely.
Then we have 
\beq
\label{hoeff}
\IP\lt[\lt|S_p -\IE S_p\rt|\geq  pt\rt]
\leq 4\exp{\lt[-{p^2 t^2\over \sum_{l=1}^p (b_l-a_l)^2}\rt]}
\eeq
for all positive values of $t$. 
\end{proposition}

We apply the Hoeffding inequality to $S_p$ with
$a_l=-1, b_l=1, l=1,...,p$ 
and
\[
t=K\sqrt{2\over p},\quad K>0
\]
to obtain 
\beq
\IP\lt[p^{-1}\lt| \sum_{k=1}^p
e^{\imath (\theta_{kj}-\theta_{kj'}) }\rt|\geq {K \sqrt{2\over p}}\rt]\leq 4e^{-K^2/2}.\label{sum2'}
\eeq

Note the dependence of  $S_p$ on $
\theta_{kj}-\theta_{kj'}$  and  the symmetry:  $|S_p
(\theta_{kj'}-\theta_{kj})|
=|S_p (\theta_{kj'}-\theta_{kj})|$. As a consequence, 
there may be $N(N-1)/2$ different values of $S_p$.
By union bound with (\ref{sum2'}), we obtain 
\beq
\IP\lt[p^{-1}\max_{j\neq j'} \lt| \sum_{k=1}^p
e^{\imath (\theta_{kj}-\theta_{kj'}) }\rt|\geq {K \sqrt{2\over p}}\rt]\leq 2N(N-1)e^{-K^2/2}\leq \delta\label{sum2}
\eeq
by (\ref{mesh}). 

Next consider the summation, denoted by $T_n$,  over the sensor locations  $l=1,...,n$ in (\ref{1.20}):
\[
T_n=\sum_{l=1}^n e^{\imath \xi_l\om(x_{j'}-x_j)/z_0}
e^{\imath \eta_l\om(y_{j'}-y_j)/z_0}.
\]
\commentout{ 
Define the random variables $A_l, B_l, l=1,...,n$, as 
\beq
\label{9.1}
A_l&=&\cos{\lt[(\xi_l(x_{j'}-x_j)+\eta_l(y_{j'}-y_j))\om/z_0\rt]}\\
B_l&=&\sin{\lt[(\xi_l(x_{j'}-x_j)+\eta_l(y_{j'}-y_j))\om/z_0\rt]}\label{9.2}
\eeq
and $S_n$ as in (\ref{9}).  

Note the dependence of  $S_n$ on $
(x_{j'}-x_j, y_{j'}-y_j)$  and  the symmetry:  $S_n
(x_{j'}-x_j, y_{j'}-y_j)
=S_n (x_j-x_{j'}, y_j-y_{j'})$. 
Furthermore, a moment of reflection reveals that thanks
to the square symmetry of the lattice 
there
are at most $N-1$ different values $|S_n|$ 
among the $N(N-1)/2$  pairs  of points. 

We apply the Hoeffding inequality to $S_n$ with
$a_l=-1, b_l=1, l=1,...,n$ 
and
\[
t=K\sqrt{2\over n},\quad K>0. 
\]
}
By the same argument  we obtain
\beqn
\IP\lt[\max_{j'\neq j} n^{-1}\lt|T_n -\IE T_n\rt|\geq K\sqrt{2\over n}\rt]
&\leq& 2N(N-1)e^{-{K^2/2}}
\eeqn
and hence 
\beq
\IP\lt[\max_{j'\neq j} {1\over n}\lt|T_n\rt|\geq a+ K\sqrt{2\over n}\rt]
&\leq& \delta,\quad \label{10.1} a=\max_{j\neq j'} {1\over n} |\IE T_n|
\eeq
by (\ref{mesh}).

 By the mutual  independence
of $\xi_l$ and $\eta_l$ we have
\beqn
a=\max_{j\neq j'}{1\over n}\lt|\IE T_n\rt|&=&\max_{j\neq j'}
{1\over n}\lt|\sum_{l=1}^n \IE \lt(e^{\imath \xi_l\om(x_{j'}-x_j)/z_0}\rt)
\IE\lt(e^{\imath \eta_l\om(y_{j'}-y_j)/z_0}\rt)\rt|\\
&=&\max_{j\neq j'} \lt|\IE \lt(e^{\imath \xi_l\om(x_{j'}-x_j)/z_0}\rt)
\IE\lt(e^{\imath \eta_l\om(y_{j'}-y_j)/z_0}\rt)\rt|
\eeqn
since $\xi_l, \eta_l, l=1,...,n$ are independently identically
distributed. 

Combining (\ref{10.1}) and  (\ref{sum2}) and noting
the independence of these two events, we obtain 
\[
\mu(\bPhi)\leq {a K\sqrt{2}\over \sqrt{p} } + {2 K^2\over \sqrt{np}}
\]
with probability at least $1-2\delta$. 

Simple calculation with the uniform distribution on the set $\cD$ given in (\ref{arr})  yields
\beq
\lt|\IE \lt(e^{\imath \xi_l\om(x_{j'}-x_j)/z_0}\rt)
\IE\lt(e^{\imath \eta_l\om(y_{j'}-y_j)/z_0}\rt)\rt| 
&=&0,\quad j'\neq j \label{20.2}
\eeq 
if (\ref{apert}) holds. 
In this case, 
\[
\mu(\bA)\leq {2}K^2/\sqrt{np}
\]
with probability $1-2\delta$.

\end{proof}

\subsection{Proof of Lemma \ref{lem:lb}: Lower bound}\label{sec5.2}
\begin{proof}
The Berry-Esseen theorem \cite{Fel}  states that the distribution of the sum  of $m$ independent and 
identically distributed zero-mean random variables 
 normalized by its standard deviation, 
differs from the unit Gaussian distribution by at most 
$C \rho/(\sigma^2\sqrt{m})$, where $\sigma^2$ and $\rho$ are respectively the 
variance and the absolute third moment of the parent distribution, and $C$ is a distribution-independent 
absolute constant which is not greater than 0.7655
\cite{Sen}.

We shall apply the Berry-Esseen theorem to 
the two summations, denoted by $S_p$ and $T_n$ respectively,  on the right hand side of (\ref{1.20}).

The complex-valued random variables involved can be
treated as $\IR^2$-valued random variables. Under
(\ref{apert}) the variance of these random variables
is $1/2$ and the absolute third moment is $4/(3\pi)$. 

Let $F_1, F_2$ be the cumulative distributions of the
real and imaginary parts  of $p^{-1/2}S_p$ and
$G_1, G_2$ the cumulative distributions of the
real and imaginary parts of $n^{-1/2} T_n$. Let $\Psi$ be
the cumulative distribution of the standard normal
random variable. We have
by the Berry-Esseen theorem
\beq
\sup_{t} |F_i(t)-\Psi(t)|&\leq& {C8\sqrt{2}\over 3\pi\sqrt{p}},\quad i=1,2\label{40}\\
\sup_{t} |G_i(t)-\Psi(t)|&\leq& {C8\sqrt{2}\over 3\pi\sqrt{n}},\quad i=1,2.\label{41}
\eeq
Since $C\leq 0.7655$, we can replace the right hand side of (\ref{40}) and (\ref{41})) by
$p^{-1/2}$ and $n^{-1/2}$ respectively for the sake
of notational simplicity. 
Hence
\beqn
|F_i(t)-F_i(-t)|&\leq& |\Psi(t)-\Psi(-t)|+{2\over \sqrt{p}}\\
|G_i(t)-G_i(-t)|&\leq& |\Psi(t)-\Psi(-t)|+{2\over \sqrt{n}}
\eeqn
$\forall t$. 
For small $t>0$ we can bound the above expressions by
\beqn
|F_i(t)-F_i(-t)|&\leq& t\sqrt{2\over \pi} +{2\over \sqrt{p}}\\
|G_i(t)-G_i(-t)|&\leq& t\sqrt{2\over \pi}+{2\over \sqrt{n}}
\eeqn
which imply 
\beqn
\IP\lt[p^{-1/2} |S_p|\leq t\sqrt{2}\rt] &\leq & 2t\sqrt{2\over \pi} +{4\over \sqrt{p}}\\
\IP\lt[n^{-1/2} |T_n|\leq t\sqrt{2}\rt] &\leq&  2t\sqrt{2\over \pi} +{4\over \sqrt{n}}
\eeqn
and consequently 
\beq
\IP\lt[{ |S_pT_n|\over np} \geq  {2t_1t_2\over\sqrt{{np}}} \rt] \geq \lt(1-2t_1\sqrt{2\over \pi} -{4\over \sqrt{p}}\rt)\lt(1-2t_2\sqrt{2\over \pi} -{4\over \sqrt{n}}\rt)\label{42}
\eeq
which is what we want to prove.

\end{proof}
\section{Average coherence bound: proof of Lemma \ref{lem2}}\label{sec5}

\begin{proof}
Write 
\beqn
\nu(\bA)
&=&{1\over N-1} \max_{j'} \lt|\sum_{l=1}^n\sum_{k=1}^p
\sum_{j\neq j'}\phi^*_{ij'} \phi_{ij}\rt|,\quad i=(k-1)n+l
\eeqn
and  consider the sums   over $k$ and $j$ simultaneously with
a fixed $j'$ and fixed $n$ sensor locations. 
This is a summation of $p\cdot N$ independent random variables  $\phi_{ij'} \phi_{ij}$ each bounded by $n^{-1}p^{-1}$ in absolute value.
Note that 
\beq
\label{35}
\IE_{\xi,\eta} \lt[ \phi^*_{ij'} \phi_{ij}\rt]=0,\quad \forall j,j', i
\eeq
since $\theta_{kj}$ are uniformly distributed in $[0,2\pi]$. 
Applying Hoeffding inequality
with 
\[
t={c_1\over (N-1)^{1/2} p^{3/2} n},\quad c>0
\]
 we have
\beq
\label{33}
\IP_{\xi,\eta}\lt[{1\over N-1}\lt|\sum_{k=1}^p
\sum_{j\neq j'}\phi_{ij'} \phi_{ij}\rt|\geq {c_1\over(N-1)^{1/2} p^{1/2}n} \rt]\leq 4 e^{-c_1^2}
\eeq
where $\IP_{\xi,\eta}$ is the probability conditioned
on fixed  $\xi=(\xi_j), \eta=(\eta_j)\in \IR^n$. 
\commentout{
Maximizing over the $m$ possible sensor locations, we obtain 
\[
\IP\lt[ \max_{j} \lt|\sum_{k=1}^p
\sum_{j\neq j'}\phi_{ij'} \phi_{ij}\rt|\geq {\sqrt{2} c\over pn^{3/2}} \rt]\leq 4m e^{-c^2 (N-1)/(2np)}.
\]
}
In analyzing the sum over $l=1,...,n$ we shall restrict to
the event 
\[
\cA=\lt\{\Theta=[\theta_{kj}]: {1\over N-1}\lt|\sum_{k=1}^p
\sum_{j\neq j'}\phi_{ij'} \phi_{ij}\rt|<
 {c_1\over np^{1/2}(N-1)^{1/2}} \,\,\hbox{for almost all  sensor locations}. \rt\}
\]
Since there are at most $N$ possible sensor locations, by (\ref{33}) 
\beq
\label{36}
\IP (\cA^c)\leq 4N e^{-c_1^2}
\eeq
where $\cA^c$ denotes the complement of $\cA$. 

Let
\[
Z_{j'l}={1\over N-1}\sum_{k=1}^p
\sum_{j\neq j'}\phi^*_{ij'} \phi_{ij} 
\]
and 
 $\IE_\cA$ is the  expectation conditioned on the event  $\cA$. 

We proceed with the following estimate
\beq
\nn\lefteqn{\IP\lt[ \max_{j'}\lt|\sum_{l=1}^n\lt(Z_{j'l} -\IE_\cA Z_{j'l}\rt)\rt|\geq 
{c_2\over \sqrt{np(N-1)}} \rt]}\\
\nn&=& \IP_\cA\lt[ \max_{j'}\lt|\sum_{l=1}^n\lt(Z_{j'l}-\IE_\cA Z_{j'l}\rt)\rt|\geq 
{c_2 \over \sqrt{np(N-1)}}  \rt]\IP (\cA)\nn\\
&&+\IP_{\cA^c}\lt[ \max_{j'}\lt|\sum_{l=1}^n\lt(Z_{j'l}-\IE_\cA Z_{j'l}\rt)\rt|\geq {c_2 \over \sqrt{np(N-1)}}\rt]\IP (\cA^c)\nn\\
&\leq&  \IP_\cA\lt[ \max_{j'}\lt|\sum_{l=1}^n\lt(Z_{j'l}-\IE_\cA Z_{j'l}\rt)\rt|\geq {c_2 \over \sqrt{np(N-1)}} \rt]+ 4N e^{-c_1^2},\quad c_1, c_2>0\label{37}
\eeq
by (\ref{36}) where $\IP_\cA$ and
$\IP_{\cA^c}$ are respectively  the probabilities  conditioned on the events $\cA$ and $\cA^c$.

Applying  Hoeffding's inequality 
with 
\[
t={c_2\over p^{1/2} (N-1)^{1/2} n^{3/2}}
\]
 to estimate the first term on the right hand side of (\ref{37}), we obtain
\[
 \IP_\cA\lt[\lt|\sum_{l=1}^n\lt(Z_{j'l}-\IE_\cA Z_{j'l}\rt)\rt|\geq {c_2 \over \sqrt{np(N-1)}}\rt]
 \leq 4 e^{-c_2^2/(2c_1)^2}. 
\]
Maximizing over $j'=1,...,m$ and using the union bound we then arrive at 
\beq
\label{38}
\IP_\cA\lt[ \max_{j'}\lt|\sum_{l=1}^n\lt(Z_{j'l}-\IE_\cA Z_{j'l}\rt)\rt|\geq {c_2\over \sqrt{np(N-1)}}\rt]
 \leq 4N e^{-c_2^2/(2c_1)^2}. 
\eeq

Using (\ref{37})  and (\ref{38}) with
\[
c_2=c\sqrt{N-1\over np}, \quad c_2=2{c_1^2},\quad c>0
\]
we have
\beq
\nn{\IP\lt[\max_{j'}\lt|\sum_{l=1}^n\lt(Z_{j'l}-\IE_\cA Z_{j'l}\rt)\rt|\geq 
{c\over {np}} \rt]}
&\leq& 8N e^{-{c\over 2}\sqrt{{N-1\over np}}},\quad c>0, 
\eeq
which is what we set out to prove. 

Note that
\[
\IE_\Theta Z_{j'l} = {1\over n} \IE \lt(e^{\imath \xi_l\om(x_{j'}-x_j)/z_0}\rt)
\IE\lt(e^{\imath \eta_l\om(y_{j'}-y_j)/z_0}\rt),\quad \forall j'=1,...,m,\quad l=1,...,n,\quad j'\neq j
\]
where $\IE_\Theta$ is the expectation conditioned on $\Theta=(\theta_{kj})\in \IC^{p\times n}$.
If
\[
{1\over \rho}={A\ell\over \lambda z_0}\in \IN
\]
then 
\[
\IE_\Theta Z_{j'l}=0,\quad \forall j'=1,...,m,\quad l=1,...,n
\]
and hence
\[\IE_\cA Z_{j'l}=0,\quad \forall j'=1,...,m,\quad l=1,...,n.
\]

\end{proof}

\section{Operator norm bound: proof of Lemma \ref{lem3}}
\label{sec:thm6}

\begin{proof}
It suffices to show that
the matrix $\bA$ satisfies
\beq
\label{gram}
\|{np\over N}\bA \bA^{*}-\bI_{np}\|_2<1
\eeq
where $\bI_{np}$ is the $np\times np$ identity matrix
with the corresponding probability bound.  Since the diagonal elements of ${np\over N}\bA\bA^{*}$ are
unity,
 (\ref{gram}) would in turn
follow from 
\beq
\label{19'}
\mu \lt(\bA^{*}\rt)< {1\over np-1} 
\eeq
by the Gershgorin circle theorem.

The pairwise coherence  has the form
\beq
\nn
{np\over N}\sum_{j=1}^N \phi_{ij}\phi^{*}_{i'j}
&=&{1\over N}  e^{\imath \om(\xi_{l}^2+\eta_{l}^2+\xi_i^2-\xi_{l'}^2-\eta_{l'}^2)/(2z_0)}\sum_{j=1}^N e^{\imath \om x_j (\xi_{l'}-\xi_{l})/z_0} e^{\imath \om y_j (\eta_{l'}-\eta_{l})/z_0}e^{\imath (\theta_{kj} -\theta_{k'j})}.
\eeq
There are two cases: (i) $k\neq k'$, (ii) $k=k', l\neq l'$.

For case (i), $\theta_{kj}-\theta_{k'j}$ are independent random
variables for $j=1,...,N$. Applying Hoeffding inequality to 
\[
Z_N\equiv \sum_{j=1}^N e^{\imath \om x_j (\xi_{l'}-\xi_{l})/z_0} e^{\imath \om y_j (\eta_{l'}-\eta_{l})/z_0}e^{\imath (\theta_{kj} -\theta_{k'j})}
\]
we obtain
\beq
\IP\lt[{1\over N} \lt| Z_N\rt|\geq t\rt]\leq 4 e^{-N t^2}.
\eeq

Set $t=\alpha/\sqrt{N}$, we have
\beqn
\IP\lt[\lt| {np\over N}\sum_{j=1}^N \phi_{ij}\phi^{*}_{i'j}\rt|\geq {\alpha \over\sqrt{ N}}\rt]\leq 4 e^{ -\alpha^2}
\eeqn
and thus
\beq
\IP\lt[\sup_{k\neq k'\atop \forall l,l'} \lt| {np\over N}\sum_{j=1}^N \phi_{ij}\phi^{*}_{i'j}\rt|\geq {\alpha \over\sqrt{ N}} \rt]\leq 2n^2p(p-1)e^{ -\alpha^2}\label{45}
\eeq
by the union bound.

For case (ii), $\theta_{kj}-\theta_{k'j}=0$ and $Z_N$ becomes a geometric
series 
\beqn
Z_N&=&{e^{\imath \om(\xi_{l'}-\xi_{l})(x_1+\sqrt{N}\ell)/z_0}-e^{\imath \om(\xi_{l'}-\xi_{l})x_1/z_0}\over 1-e^{\imath \om(\xi_{l'}-\xi_{l})\ell/z_0}}\times
{e^{\imath \om(\eta_{l'}-\eta_{l})(y_1+\sqrt{N}\ell)/z_0}-e^{\imath \om(\eta_{l'}-\eta_{l})y_1/z_0}\over 1-e^{\imath \om(\eta_{l'}-\eta_{l})\ell/z_0}}.
\eeqn

Thus,
\beqn
{np\over N} \lt|\sum_{j=1}^N \phi_{i j}\phi^{*}_{i'j}\rt|&\leq& {1\over N}
\lt|{\sin{\om\ell\sqrt{N}(\xi_{l'}-\xi_{l})\over 2z_0} \over \sin{\om\ell (\xi_{l'}-\xi_{l})\over 2z_0}}\rt| \cdot\lt|{\sin{\om\ell\sqrt{N}(\eta_{l'}-\eta_{l})\over 2z_0} \over \sin{\om\ell (\eta_{l'}-\eta_{l})\over 2z_0}}\rt|.  
\eeqn

Let 
\beqn
\label{view2}
\kappa=\min_{l\neq l'}\min_{j\in \IZ }\lt\{\lt|{ \ell (\xi_{l'}-\xi_{l})\over  \lambda z_0}-j\rt|, \lt|{ \ell (\eta_{l'}-\eta_{l})\over \lambda z_0}-j\rt|\rt\}. 
\eeqn

Clearly $\kappa$ is nonzero with probability one.
For $l\neq l'$ the probability density functions  (PDF) for the random variables
\[
{ \ell (\xi_{l'}-\xi_{l})\over  \lambda z_0},\quad
{ \ell (\eta_{l'}-\eta_{l})\over \lambda z_0}
\]
are either  the symmetric triangular distribution or
its self-convolution supported
on $[- 2\rho^{-1},  2\rho^{-1}]$. 
In either case, their PDFs are bounded by $\rho$. 
Hence the probability that $\{\kappa >\beta\}$ for small $\beta>0$ is larger than 
\[
(1-2\rho\beta)^{n(n-1)/2} >1-\beta\rho n(n-1)
\] 
where the exponent counts the number of distinct  unordered  pairs $(l,l')$. Note that the above analysis is independent
of $k=k'$. Since $\sin\theta\geq \theta, \forall \theta\in [0,\pi/2]$ we
have that
\beq
\IP\lt[\sup_{k=k'\atop l\neq l'} {np\over N} \lt|\sum_{j=1}^N \phi_{ij}\phi^{*}_{i'j}\rt|\geq  {\pi^2\over 4 N\beta^2}\rt]\leq \beta\rho n(n-1).\label{49}
\eeq

Setting  
\beq
\label{201}
\max\lt\{{\alpha\over \sqrt{N}}\rt\} < {\pi^2\over 4N\beta^2}= {1\over np-1}
\eeq
and using  (\ref{45}) and (\ref{49}) we have
 \beq
\IP\lt[\sup_{i\neq i'} {np\over N} \lt|\sum_{j=1}^N \phi_{ij}\phi^{*}_{i'j}\rt|\geq {1\over np-1}\rt]\leq \beta\rho n(n-1)+2n^2p(p-1) e^{-\alpha^2}.\label{50}
\eeq

As a consequence, 
 \beqn
\IP\lt[\sup_{i\neq i'} {np\over N} \lt|\sum_{j=1}^N \phi_{ij}\phi^{*}_{i'j}\rt|\geq  {1\over np-1}\rt]<  \rho n(n-1){\pi \over 2} \sqrt{np-1\over N}+2n^2p(p-1) e^{-{N\over (np-1)^2}}
\eeqn
by maximizing the right hand side of (\ref{50}) under the constraint (\ref{201}). 
\end{proof}

\commentout{
\subsection{Average coherence bound}

Average coherence is defined as \cite{BCJ}
\[
\nu(\bPhi)={1\over N-1}
\max_{j'} \lt| \sum_{j\neq j'} \Phi_{j'}^*\phi_j\rt|
\]
which can be written as
\[
\nu(\bPhi) ={m\over N-1} \max_{j} \lt|\Phi^*_{j}\bar \Phi -{1\over N}\rt|
\]
with 
\[
\bar \Phi={1\over N} \sum_{j'=1}^N \Phi_{j'}
\] 
 being the row-averaged  column vector. 
 \begin{proposition}\label{prop1}
 For any set of  sensor locations
 \beq
\IP_{\xi,\eta}\lt[ |\bar \Phi|_\infty < K\sqrt{2\over Nnp}\rt]> 1-2\delta {np\over N}
\eeq
where $\IP_{\xi,\eta}$ is the conditional probability
given the sensor locations. 
 \end{proposition}
\begin{proof}
Consider
 \beq
 \bar\Phi_{i} &=& {1\over N\sqrt{np}}\sum_{j=1}^N e^{\imath \om|x_j-\xi_l|^2/(2z_0)}e^{\imath \om|y_j-\eta_l|^2/(2z_0)}e^{\imath \theta_{kj}},\quad  i=(k-1)n+l 
 \eeq

Let  
\beq
\label{9.1'}
X_{j}&=&\cos{\lt[{\om\over 2 z_0}\lt((\xi_l-x_j)^2+(\eta_l-y_j)^2\rt)+\theta_{kj}\rt]}\\
Y_{j}&=&\sin{\lt[{\om\over 2 z_0}\lt((\xi_l-x_j)^2+(\eta_l-y_j)^2\rt)+\theta_{kj}\rt]} \label{9.2'}
\eeq
and their respective sums
\beqn
\label{9'}
S_i=\sum_{j=1}^N X_j,\quad T_N=\sum_{j=1}^N Y_j.
\eeqn
Note that $\IE S_N=\IE T_N=0$ for all $\xi_l, \eta_l, k$  by the definition of $\theta_{kj}$. 
Hence 
\beq
\label{sum'}
{1\over N} \lt|S_N+iT_N\rt|\leq {1\over N}\sqrt{\lt|S_N\rt|^2
+\lt|T_n\rt|^2}.
\eeq

Applying  the Hoeffding inequality with 
\[
t=K/\sqrt{N}
\]
we  obtain 
\beq
\label{hoeff2'}
\IP_{\xi,\eta}\lt[m^{-1}\lt|S_N\rt|\geq K/\sqrt{N}\rt]
&\leq& 2e^{-{K^2/2}}\\
\IP_{\xi,\eta}\lt[m^{-1}\lt|T_N\rt|\geq K/\sqrt{N}\rt]
&\leq& 2e^{-{K^2/2}}\label{hoeff22'}
\eeq
and hence by (\ref{sum'})
\beq
\IP_{\xi,\eta}\lt[ |\bar \phi_i| \geq K\sqrt{2\over Nnp}\rt]\leq 4e^{-K^2/2}.
\eeq

To estimate $\max_i |\bar \phi_i| $ we use the union bound to obtain
\beqn
\IP_{\xi,\eta}\lt[ \max_i |\bar \phi_i| \geq K\sqrt{2\over Nnp}\rt]\leq 4np e^{-K^2/2}
\eeqn
which is what we wanted to prove.

\end{proof}

To estimate $\Phi_{j}^*\bar \Phi-{1\over N}$ we write
\[
\Phi_{j}^*\bar \Phi -{1\over N}=\sum_{k=1}^p\sum_{l=1}^n
\lt[\phi_{ij}^*\bar \Phi_{i}-{1\over Nnp}\rt]
\]
with $ i=(k-1)n+l$. 

First consider the summation over the sensor locations $l=1,...,n$. 
From the selections of the sensor locations, this is
a sum of independent random variables for a fixed
set of $\{\theta_{kj}: k=1,...,p; j=1,...,m\}$ since
the $i$-th summand depends only on the $i$-th sensor
location. Note
that
\[
\IE_{\Theta} \lt[\phi^*_{ij}\bar \Phi_{i}\rt]-{1\over Nnp}
={1\over N} \IE_\Theta \lt[\phi^*_{ij}\sum_{j'\neq j} \phi_{ij'}\rt]
= {1\over N} \sum_{j'\neq j}\IE_\Theta \lt[\phi^*_{ij} \phi_{ij'}\rt]=0
\]
by the aperture condition (\ref{apert}) where
$\IE_\Theta$ is the conditional expectation given
$\Theta=[\theta_{kj}]$. 
Applying 
the Hoeffding inequality with
\[
t={K^2\sqrt{2}\over\sqrt{mp^2n^3}}
\]
conditioned on
the event, $E$,  in Proposition \ref{prop1}
we obtain 
\beq
\IP_\Theta\lt[ \lt|\sum_{l=1}^n
\phi^*_{ij}\bar \Phi_{i}-{1\over Np}\rt|\geq {K^2\sqrt{2}\over p \sqrt{mn}}| E\rt]\leq 4e^{-K^2/2}.  
\eeq
Maximizing over $j$ using the union bound as before  we 
have 
\beqn
\IP_\Theta\lt[ \max_{j} \lt|\sum_{l=1}^n
\phi^*_{ij}\bar \Phi_{i}-{1\over Np}\rt|\geq {K^2\sqrt{2}\over p \sqrt{mn}}\rt]\leq 4(N-1)e^{-K^2/2}<2\delta
\eeqn
or
\beq
\label{42}
\IP_\Theta\lt[ \max_{j} \lt|\sum_{l=1}^n
\phi^*_{ij}\bar \Phi_{i}-{1\over Np}\rt|< {K^2\sqrt{2}\over p \sqrt{mn}}\rt]>1-2\delta.
\eeq

Now let us turn to the summation over independent random
illuminations $k=1,...,p$.   First note
\[
\IE_{\xi,\eta} \sum_{l=1}^n\lt[\phi^*_{ij}\bar \Phi_{i}-{1\over Nnp}\rt]=0
\]
by the $2\pi$-uniform distribution of $\theta_{kj}$ where
$\IE_{\xi,\eta}$ is the conditional expectation given
the sensor locations. Applying Hoeffding inequality with
(\ref{42}) we obtain 
\beq
\label{43}
\IP\lt[ \max_{j} \lt|\sum_{k=1}^p\sum_{l=1}^n
\lt[\phi^*_{ij}\bar \Phi_{i}-{1\over Nnp}\rt]\rt|< {K^3\sqrt{2}\over \sqrt{mnp}}\rt]>1-2\delta.
\eeq
}

\commentout{
 \subsection{Operator norm estimate}
 We follow the idea in \cite{RV} which is to estimate
 \[
 \IE \lt\|\bI_{np} -  {np\over N}\bPhi\bPhi^*\rt\|.
 \]
 First note that 
 \beq
 \bPhi\bPhi^*= \sum_{j=1}^N
 \Phi_j\otimes\Phi_j^*
 \eeq
 where $\Phi_j$ is the $j$-th column of $\bPhi$. 
 Moreover, 
 \beq
 \label{31}
 {np\over N} \sum_{j=1}^N
 \IE\lt[\Phi_j\otimes\Phi_j^*\rt]=\bI_{np}
 \eeq
 which can be derived as follows.
 We have 
 \beq
\lefteqn{ \IE\lt[\phi_{ij}\otimes\phi_{i'j}^*\rt]}\label{30}\\ &=& 
{1\over np} \IE\lt[e^{\imath \om x_j (\xi_{l'}-\xi_l)/(2z_0)}
e^{\imath \om y_j (\eta_{l'}-\eta_l)/(2z_0)}
e^{\imath \om(\xi_l^2-\xi_{l'}^2)/(2z_0)}
e^{\imath \om(\eta_l^2-\eta_{l'}^2)/(2z_0)}
 e^{\imath \theta_{kj}} e^{-\imath \theta_{k'j}}\rt]\nn
 \eeq
 with
 \beq
 i=(k-1)n+l ,\quad i'=(k'-1)n+l'.\label{ind}
 \eeq
 If $k\neq k'$, then the expectation in (\ref{30}) is zero
 due to the independence of $\theta_{kj}$ and  $\theta_{k'j}$.
 If $l\neq l'$, then the expectation is again zero
 due to the independence of $(\xi_l, \eta_l)$ and $(\xi_{l'},\eta_{l'})$. If $i=i'$, then the expectation is $(np)^{-1}$. 
 Hence the identity (\ref{31}). 
 
 Second, observe that for any given  the sensor locations 
the columns of $\bPhi$ are independent due to random illuminations. However,
they are not independent in the joint probability space
of sensors and illuminations. Denote by $\IE_{\xi,\eta}$  the
conditional expectation given the sensor locations
$\xi=(\xi_1,...,\xi_n), \eta=(\eta_1,...,\eta_n)$. 
We have 
\[
\IE_{\xi,\eta}\lt[f(\Phi_j) g(\Phi_{j'})\rt]
=\IE_{\xi,\eta} [f(\Phi_j)]\IE_{\xi,\eta} [g(\Phi_{j'})],\quad j\neq j'
\]
for all bounded measurable functions $f,g$. 

\begin{proposition}
Let $A$ and $B$ be two random vectors in $\IC^{np}$ that
are conditionally independent w.r.t. $\IE_{\xi,\eta}$.
If $\IE B=0$, then
\[
\IE \|A+B\|_2 \geq \IE \|A\|_2.
\]

\end{proposition}
\begin{proof}
We have the calculation 
\beqn
\IE \|A+B\|_2 =\IE \lt[\IE_{\xi,\eta} \|A+B\|_2\rt]
&=&\IE \lt[\IE_{\xi,\eta} \|A+\IE_{\xi,\eta} B+ (B-\IE_{\xi,\eta} B)\|_2\rt].
\eeqn
By  the conditional independence of
$A$ and $B-\IE_{\xi,\eta}$, partial integration with respect
to $B$ using Fubini's theorem and Jensen's inequality,  
\beqn
\IE_{\xi,\eta} \|A+\IE_{\xi,\eta} B+ (B-\IE_{\xi,\eta} B)\|_2
&\geq&\IE_{\xi,\eta} \|A+\IE_{\xi,\eta} B\|_2.
\eeqn
Hence
\beqn
\IE \|A+B\|_2 &\geq &\IE\lt[
\IE_{\xi,\eta} \|A+\IE_{\xi,\eta} B\|_2\rt]
=\IE \|A+\IE_{\xi,\eta} B\|_2. 
\eeqn

\end{proof}
}

  \commentout{
\section{Localized extended objects}
In this section we extend the above result to
the case of extended  objects that are localized in  space. 

 We 
represent such objects  by piecewise constant interpolation from grid points.  
Suppose that the function $\tau(\br)=\nu(\br)u(\br)$ is continuous 
and has  a compact support. 
Consider the interpolation from the lattice $\cL$
\beq
\label{533}
\tau_{\ell}(x,y)=\sum_{\bq\in \cL/\ell}
 g({(x,y)\over \ell}-\bq ) \tau(\ell\bq)
\eeq
where $g$ is the interpolation element. 

 Let $\bq=(q_1,q_2)$ and $l=(q_1-1)\sqrt{N}+ q_2$. 
Define the vector $X=(X_l)\in \IC^N$ with $X_l=\tau(\ell\bq)$ as before. Immediately after the object plane we
place a parabolic lens which has the effect of multiplying
the object with
the quadratic phase factor
\[
e^{-\imath \om (x^2+y^2)/(2z_0)}.
\]
This lens is introduced here to simplify our analysis. In practice, 
the lens should not be needed and should have
a negligible effect on performance.

After proper normalization of the scattered field data we can write
the data vector $Y$ 
in the form (\ref{u1}) with  the sensing matrix elements 
\beqn
\label{535}
\phi_{jl}&=&{2z_0 e^{-\imath \om z_0} \over \ell^2 \sqrt{n}  \hat g((\xi_j,\eta_j)\om \ell/z_0)}
\int g({(x,y)\over\ell} -\bq)e^{-\imath \om (x^2+y^2)/(2z_0)}\Gp((\xi_j,\eta_j,0), (x,y,z_0)) dx dy
\eeqn
where $\hat g$ is the Fourier transform of $g$. 
In this setting
the error term $E$ in (\ref{u1}) also includes the  discretization error. With some algebra this leads to the simple expression
\beqn
\phi_{jl}&=& {1\over \sqrt{n}} e^{\imath \om (\xi_j^2+\eta_j^2)/(2z_0)}
e^{-\imath \om \xi_j q_1\ell/z_0} e^{-\imath \om \eta_j q_2\ell/z_0}
\eeqn
which has the form (\ref{sand}) with $\bD_2=\bI$ and 
the same $\bD_1$ and $\bPsi$ after
identifying $x_l=q_1\ell, y_l=q_2\ell$.

Exactly the same argument now extend  Theorem \ref{thmA} 
to the case of localized extended objects if 
we can show that
 the discretization error $E_d$ can be  made arbitrarily small,
e.g. 
$\|E_d\|_2\leq \ep'$ for any given $\ep'>0$. 

 First, by the inequality $\|E_d\|_2\leq \|E_d\|_\infty \sqrt{n}$
it suffices to have $\|E_d\|_\infty \leq \ep /\sqrt{n} $. Now
consider the transformation
$\cT$, defined by 
\beqn
 (\cT v)_j &=&{2z_0 e^{-\imath \om z_0} \over  \sqrt{n}  \hat g((\xi_j,\eta_j)\om\ell/z_0)}
\int v(x,y) e^{-\imath \om (x^2+y^2)/(2z_0)}\Gp((\xi_j,\eta_j,0), (x,y,z_0)) dx dy
\eeqn
 from  the space of continuous functions supported on  $[\ell, m\ell]^2$ to
$\IC^n$. By definition
  \[
 E_d= \cT \tau -\cT \tau_{\ell}
 \]
 we have
 \[
 \|E_d\|_\infty\leq {2z_0 \|\tau-\tau_{\ell}\|_1\over \sqrt{n}  \min_{j}|\hat g((\xi_j,\eta_j)\om\ell/z_0)|}.
\]
and hence 
\beqn
\|E_d\|_\infty\leq {2z_0\over \sqrt{n}}  {\|v-v_{\ell}\|_1}
\|\hat g^{-1}\|_{L^\infty([-\pi, \pi]^2)}. 
\eeqn
where $\|\cdot\|_{L^\infty([-\pi, \pi] ^2)}$ denotes
the $L^\infty$-norm of functions defined on $[-\pi, \pi]^2$. 
 
 \begin{theorem} \label{thm:loc-ext}
Consider imaging method as described above. Assume  (\ref{75})  and
\beq
\label{disc-err}
 {\|v-v_{\ell}\|_1}
\leq  {\ep' \over 2z_0 \|\hat g^{-1}\|_{L^\infty([-\pi, \pi]^2)}}. 
\eeq
Suppose that (\ref{77}) 
 holds with $\sigma=2s$ and any $\delta<\sqrt{2}-1$.   Then 
the bound (\ref{101})
holds  true  with probability at least $1-\gamma$.
\end{theorem}

\begin{remark}
 For example, consider the indicator function  $g$ on
 the unit square $[-{1\over 2}, {1\over 2}]^2$. The resulting
 interpolation  is the
 piecewise constant approximation. Then
 \[
 \hat g(k_1, k_2)={2\over \pi}\cdot {\sin{k_1\over 2}\over k_1} \cdot
 {\sin{k_2\over 2}\over k_2}
 \]
We obtain the condition
 \beq
 \label{piece}
  {\|v-v_{\ell}\|_1}\leq
 {\ep'\over \pi^3 z_0}
 \eeq
 which can be realized for sufficiently small $\ell$ since piecewise constant interpolation is accurate up to  first-order:
  \[
   {\|v-v_{\ell}\|_1}\leq c \ell
   \]
   with a constant $c$ independent of $\ell$.

\end{remark}

}

\section{Numerical simulations }\label{sec:num}

  \begin{figure}[t]
\begin{center}
\includegraphics[width=0.8\textwidth]{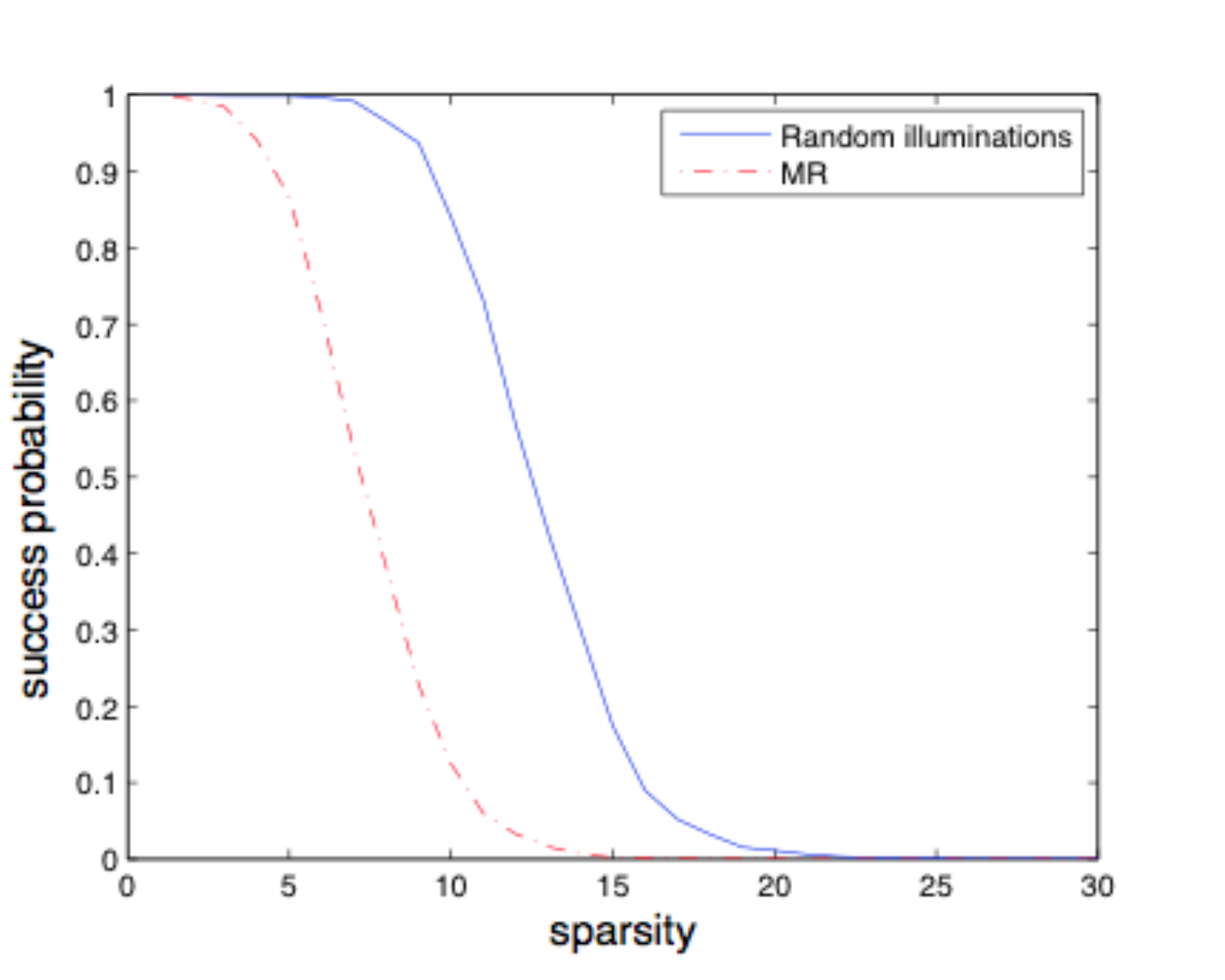}
\end{center}
\caption{The Lasso performance comparison between RI with $n=11, p=6$ and MR with $n=11$. The vertical axis is for the success probability and the horizontal axis is for the number of objects.
The success probability is estimated from 1000 independent trials.}
\label{fig1}
\end{figure}

  \begin{figure}[t]
\begin{center}
\includegraphics[width=0.8\textwidth]{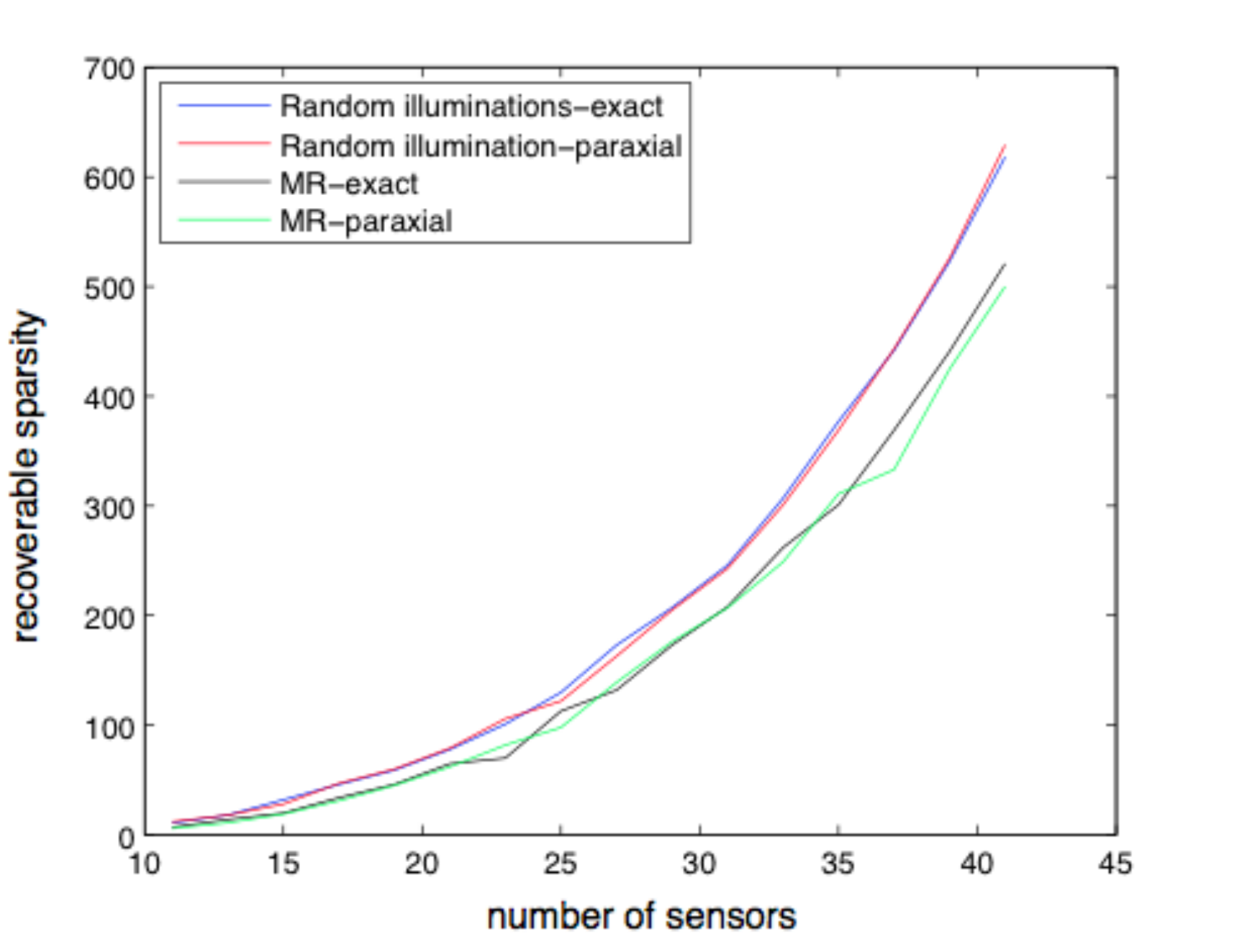}
\end{center}
\caption{The numbers  of recoverable (by the Lasso) objects 
for  RI with $p=(n+1)/2$ and MR as $n$ varies. The curves indicate a quadratic behavior predicted by the theory. The difference between recoveries  with the exact and paraxial Green functions is negligible in both the RI and MR set-ups. }
\label{fig1'}
\end{figure}

We use two numerical settings: the diffraction-limited case when (\ref{apert}) is satisfied  (Figure \ref{fig1}, \ref{fig1'},  \ref{fig4}, \ref{fig5}) and the under-resolved case
when the ratio in (\ref{apert})  is smaller than unity (Figure \ref{fig6}).

For the diffraction-limited case we set $z_0=10000$ and 
 $\lambda=0.1$ for the search domain  $[-250, 250]^2$ with $\ell=10$. The targets are i.i.d. uniform random points
 in the grid with amplitudes  in the range  $[1, 2]$. 
We randomly select sensor locations from $[-50,50]^2$
with the aperture  $A=100$ satisfying  (\ref{apert}). With these parameters
\[
{(A+\ell\sqrt{N})^4\over \lambda z^3_0}\approx 1.3
\]
the condition (\ref{6-2}) is barely satisfied. 
 For the Lasso solution  we have used  the Matlab  code
 {\em Subspace Pursuit} (available at {\tt http://igorcarron.googlepages.com/cscodes}).

We use the 
true Green
function 
(\ref{green}) in the computation of  scattered waves
and in recovery the exact Green function as well as  its  paraxial approximation to construct the sensing matrix (for comparison). In other words, we allow model 
 mismatch between the forward  and inversion steps.
 
 In the first set of simulations, we compare 
 the performances of the Lasso for two
 imaging set-ups: one with  random illumination (RI) 
 and the other with {\em multi-static responses} (MR).
As Figure \ref{fig1} shows, the RI set-up has
a higher success probability than the MR set-up.
Another comparison is shown in Figure \ref{fig1'} in terms
of the number of recoverable objects over a range of
$n$. The quadratic behavior is consistent with
the prediction of (\ref{162}) and (\ref{138}). The difference between the exact and paraxial Green functions recoveries  is negligible in both the RI and MR set-ups. For a given $n$, the Lasso with the RI set-up
recovers a higher number of objects than does the Lasso
with the MR set-up.

Figure \ref{fig4} compares the performances
of the Lasso (top panel) and OST (bottom panel)
in terms of the number of recoverable objects for
a fixed $np=600$ but variable $n$. Clearly, the Lasso can recover
far more objects exactly than does the OST.
For a fixed $np$ the performance for each
method appears relatively constant over 
the whole range of $n$. For small $n$, the performance curves of both methods indicate superresolution. As noise level increases 
the Lasso performance decays  (Figure \ref{fig5}).

\commentout{
\section{MF and RIP}
Let $Y$ be the $n$-dimensional signals
received by the array of $n$ sensors. 
Let $G(\br,\ba)$ be the Green function
of the time-invariant medium 
and let $\bG$ be
the Green  vector
\beq
\label{1.0}
\bG(\br) =[G(\br,\ba_1),
G(\br,\ba_2),...,G(\br,\ba_n)]^t
\eeq
where $t$ denotes transpose.
The conventional MFP
uses the Bartlett processor with the
 ambiguity surface
\beq
\label{Bar}
B(\br)={\bG^*(\br)YY^*\bG(\br)\over
\|\bG(\br)\|_2^2}
\eeq
\cite{Tol}. 
The Bartlett processor is motivated by the
following optimization problem:
Maximize
the quantity
\beq
\label{snr}
W^* YY^* W
\eeq
subject to the constraint:
\[
W^* W=1.
\]
The solution 
\[
W=Y/\|Y\|_2
\]
is the weight vector for
the matched filter. 
In the case of one point source of amplitude $\tau_1$ located at $\bx_1$, 
\[
Y=\tau_1\bG(\br_1)
\]
hence
\beq
\label{mf1}
W={\tau_1\bG(\br_1)\over |\tau_1| \|\bG(\br_1)\|_2}.
\eeq
Extending (\ref{mf1}) to an arbitrary field point $\br$ 
by substituting $\br$ for  $\br_1$
we obtain  the Bartlett processor from (\ref{snr}).

In our setting with $G=G_p$, $\|\bG(\br)\|_2$ is a constant
independent of $\br$ in the computation domain
and hence $B(\br)$ can be interpreted as
 the intensity of the time reversed  field $\bG^*(\br) Y$ at $\br$ modulo a constant factor. Here {\em time reversal}  refers
 to phase conjugation in $\bG^*(\br) Y$. 

In the case of imaging, time reversal is implemented 
computationally and is the simplest form of MFP.
In the case of focusing or communication,  
time reversal is implemented  physically by reversing
and back-propagating 
the wave front from antennas  (see \cite{trm} and references therein). 

}

\commentout{
Now we would like to point out an interesting connection
between time reversal and RIP.

In the simulations, we set $z_0=10000$ and for the most part
 $\lambda=0.1$ for the computation domain  $[-250, 250]^2$ with $500\times 500$ grid points. The targets are i.i.d. uniform random variables
 of the grid with target amplitudes  in $[1, 2]$. 
We use the aperture  $A=100$ which,according to Theorem~\ref{thm3},  is the threshold, optimal aperture. Therefore the paraxial regime (\ref{6-2}) is enforced.

 \begin{figure}[t]
\begin{center}
\includegraphics[width=0.47\textwidth]{new-figures/tr_equal2500.pdf}
\includegraphics[width=0.47\textwidth]{new-figures/tr_random2500.pdf}
\end{center}
 \caption{Time reversal of waves from $40$  source points from 
  $2500$ antennas equally spaced (left) and
  randomly placed (right) in the aperture $[-50,50]^2$.  The red circles represent the true locations of the targets.  }
 \label{fig-tr}
 \end{figure}

  \begin{figure}[t]
\begin{center}
\includegraphics[width=0.47\textwidth]{new-figures/tr-equal121.pdf}
\includegraphics[width=0.47\textwidth]{new-figures/Passive_SP_paraxial_MFP.pdf}
\end{center}
 \caption{Time reversal of waves from $40$  source points from 
  $121$ antennas equally spaced (left) and
  randomly placed (right) in the aperture $[-50,50]^2$.  The red circles represent the true locations of the targets.  }
 \label{fig-tr'}
 \end{figure}
 }
 
 \begin{figure}[h]
\begin{center}
\includegraphics[width=0.8\textwidth]{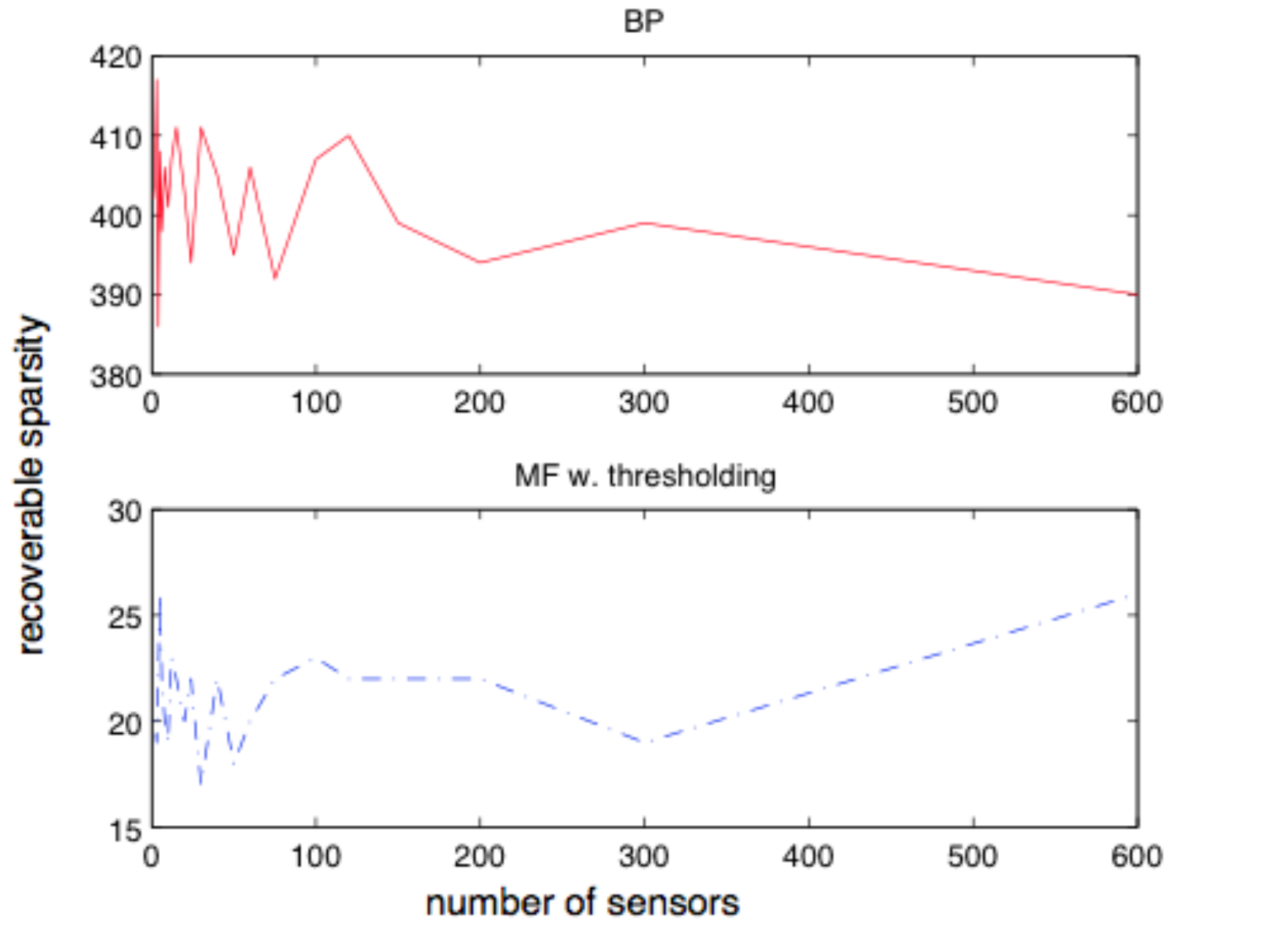}
\end{center}
\caption{The number of recoverable objects as a function of the number of sensors $n =$ 1,2,3,4,5, 6, 8, 10, 12, 15,  20, 24, 25, 30, 40, 50, 60, 75, 100, 120, 150, 200, 300, 600 with $np=600$ fixed. The top panel is for the Lasso
and the bottom panel for OST. 
The  left ends of both curves indicate superresolution. }
\label{fig4}
\end{figure}

  \begin{figure}[t]
\begin{center}
\includegraphics[width=0.80\textwidth]{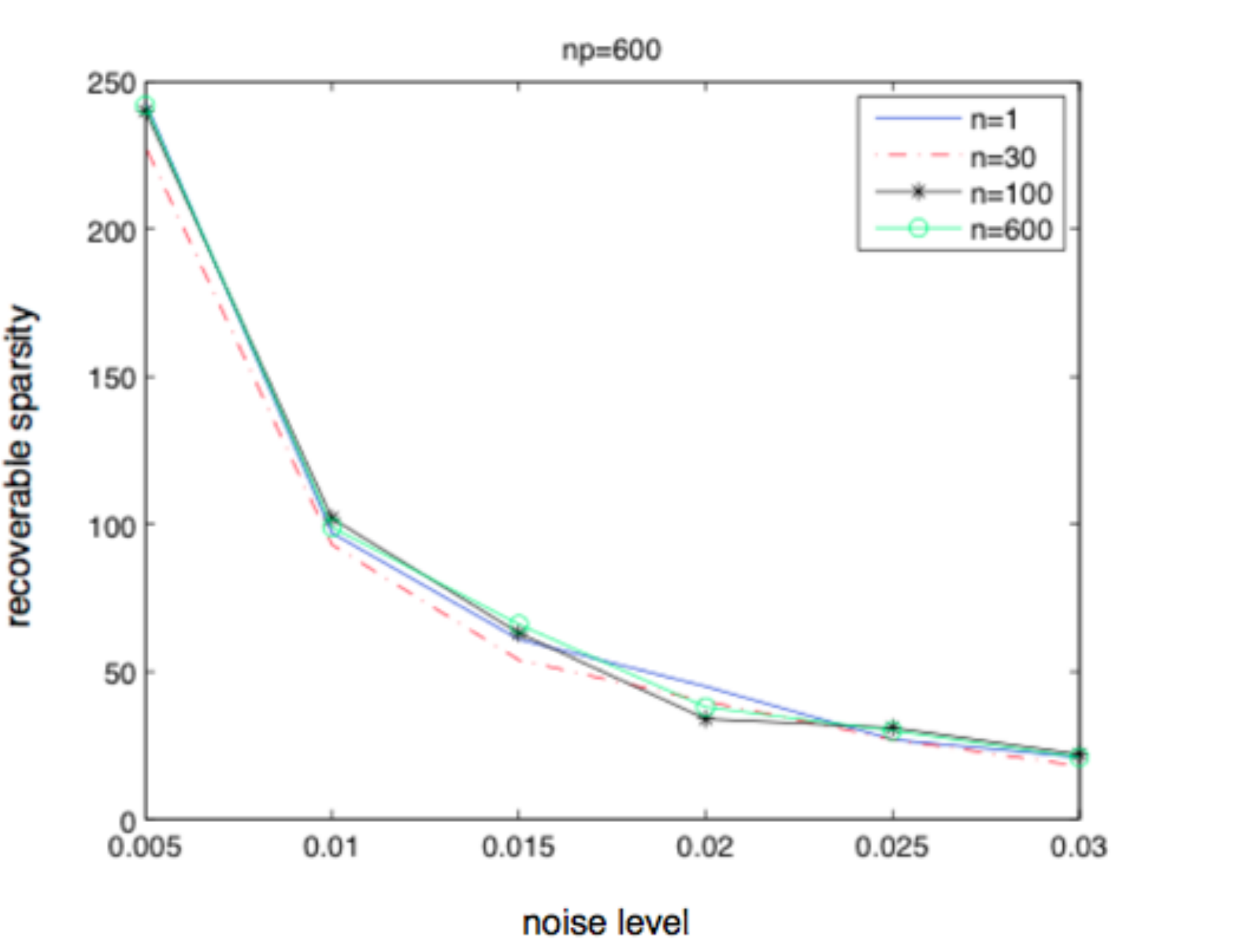}
\end{center}
\caption{Noisy recovery by the Lasso for $n=1,30,100,600$
and with $np=600$ fixed. The noise is given by the circularly random 
Gaussian noise of magnitude
$\sigma\|Y\|_2$ where $\sigma$ is the horizontal coordinate. 
Note that in this case $\IE\|E\|_2^2=np\sigma^2\|Y\|_2^2.$}
\label{fig5}
\end{figure}

To  further understand  the superresolution effect of random illumination, we consider
the set-up with $z_0=25000, \lambda=0.4$ for which
the ratio in (\ref{apert}) is $0.1$. This is an under-resolved case
whose performance  is shown in Figure \ref{fig6}. In contrast to  the
diffraction-limited case  (Figure \ref{fig4}), the number of recoverable objects in the under-resolved case
decays rapidly as $p$ decreases ($n$ increases).  To maintain high
performance in the under-resolved case, it is necessary
that $p\gg 1$. The number of recoverable objects is calculated
based on 90\% recovery of 100 independent trials.

We demonstrate   in Figures \ref{fig8}- \ref{fig9'} the performance for extended objects in the presence of  external  noise of the form
\[
{p\over \sqrt{2}} (\nu_1+ \imath \nu_2 ) { \|Y\|_2 \over \sqrt{np}},\quad
p=5\%, 20\%
\]
where $p$ is the percentage of noise in each entry of the data vector and $\nu_1,\nu_2$ are i.i.d. uniform random variables in $[0,1]$.

Figure \ref{fig8} shows the original 
$40\times 80$ pixel image (left)  and its reconstructions  (middle panel, $5\%$ noise;  right panel, $20\%$ noise) by the BPDN solver YALL1 ({\tt http://yall1.blogs.rice.edu/}) using  one sensor and 500 random illuminations while Figure \ref{fig8'} shows the results
with one illuminations and 500 randomly distributed sensors.   

Figure \ref{fig9} shows the original $70\times 70$ pixel image (left), the Shepp-Logan phantom, and its reconstructions (middle panel, $5\%$ noise; right panel, $20\%$ noise) by the total-variation minimization \cite{CL, ROF}  solver TVAL3 (http://www.caam.rice.edu/~optimization/L1/TVAL3/) using 
one sensor and  1000 random illuminations while
Figure \ref{fig9'} shows the results with 
one illumination and 1000 randomly distributed sensors. 

The  low pixel numbers
 are chosen to reduce the run time of the programs. 

For the one-illumination reconstructions (Figures \ref{fig8'} and \ref{fig9'}), the classical resolution criterion (\ref{apert}) is met. 
Note, however, that the Shepp-Logan phantom is not 
in the class of sparse extended objects analyzed  in Section \ref{sec:ext} because the object support covers more than $50\%$ of
the domain (only the gradient is sparse). As a result, the same percentage of noise represents a greater amount of
noise in the case of Shepp-Logan phantom and has a more
serious  effect on performance (Figures \ref{fig9} and \ref{fig9'}, right panels).

 \begin{figure}[h]
\begin{center}
\includegraphics[width=0.8\textwidth]{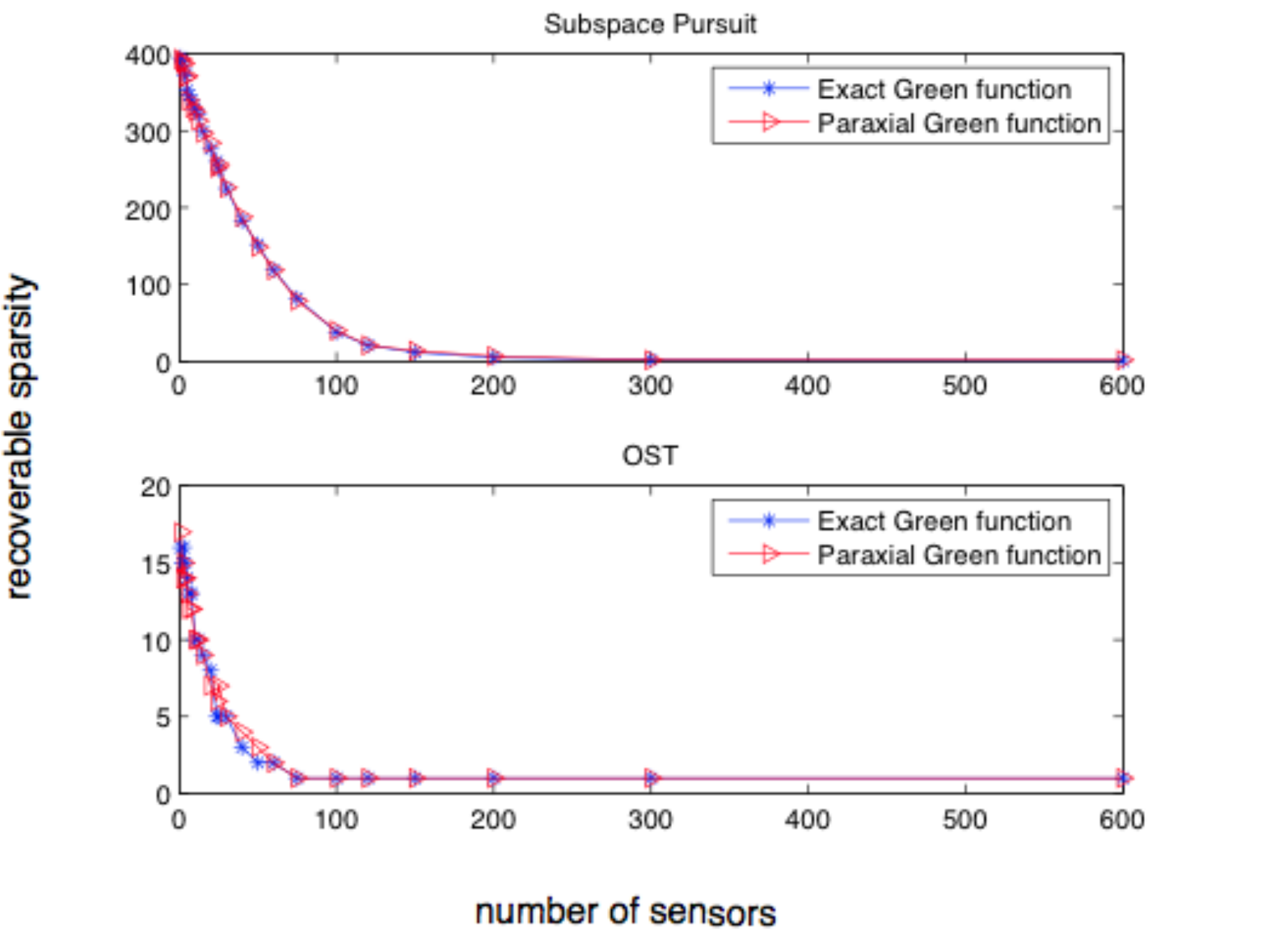}
\end{center}
\caption{The number of recoverable objects in the under-resolved case as a function of the number of sensors $n =$ 1,2,3,4,5, 6, 8, 10, 12, 15,  20, 24, 25, 30, 40, 50, 60, 75, 100, 120, 150, 200, 300, 600 with $np=600$ fixed.  }
\label{fig6}
\end{figure}

 \begin{figure}[h]
\begin{center}
\includegraphics[width=0.32\textwidth]{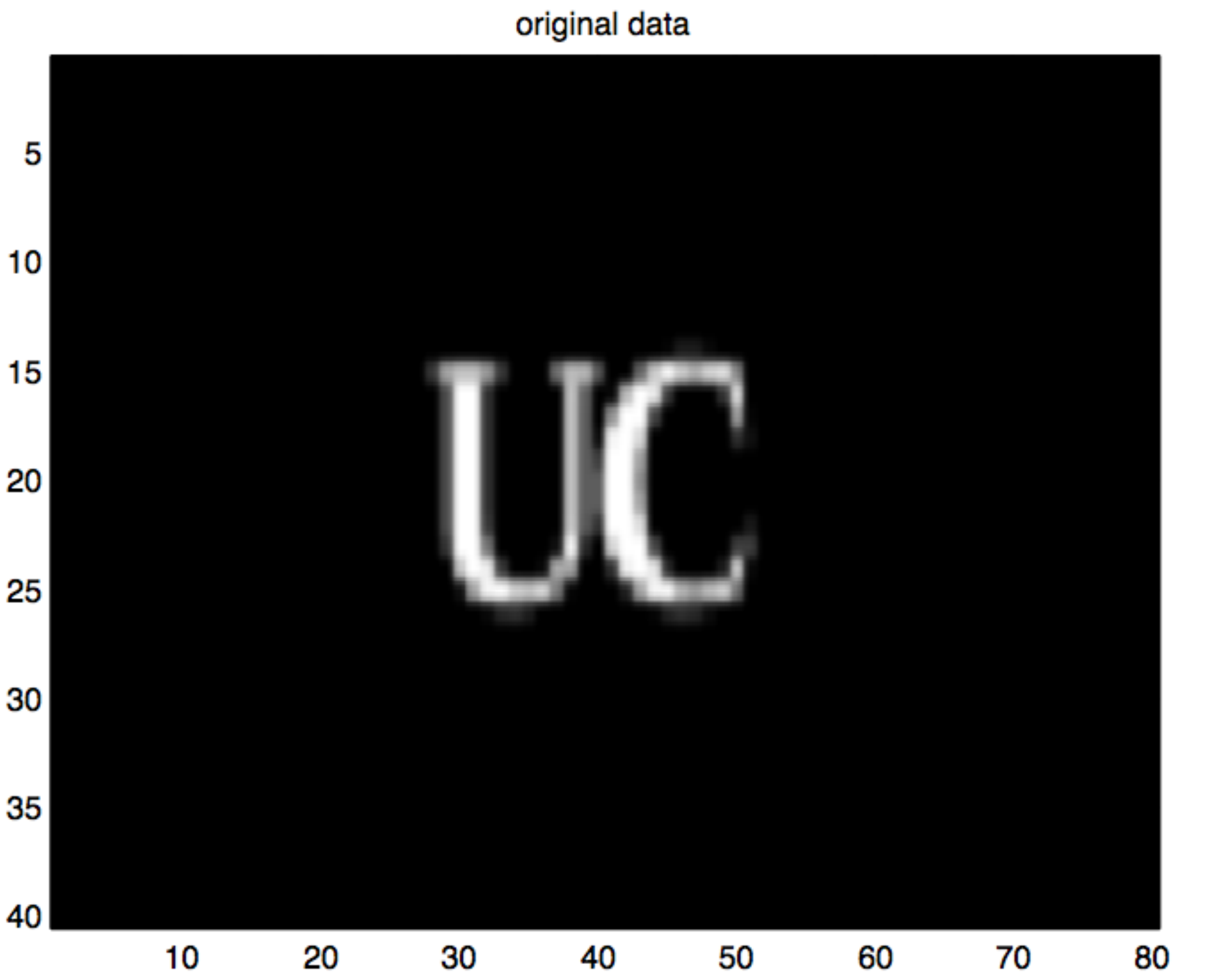}
\includegraphics[width=0.32\textwidth]{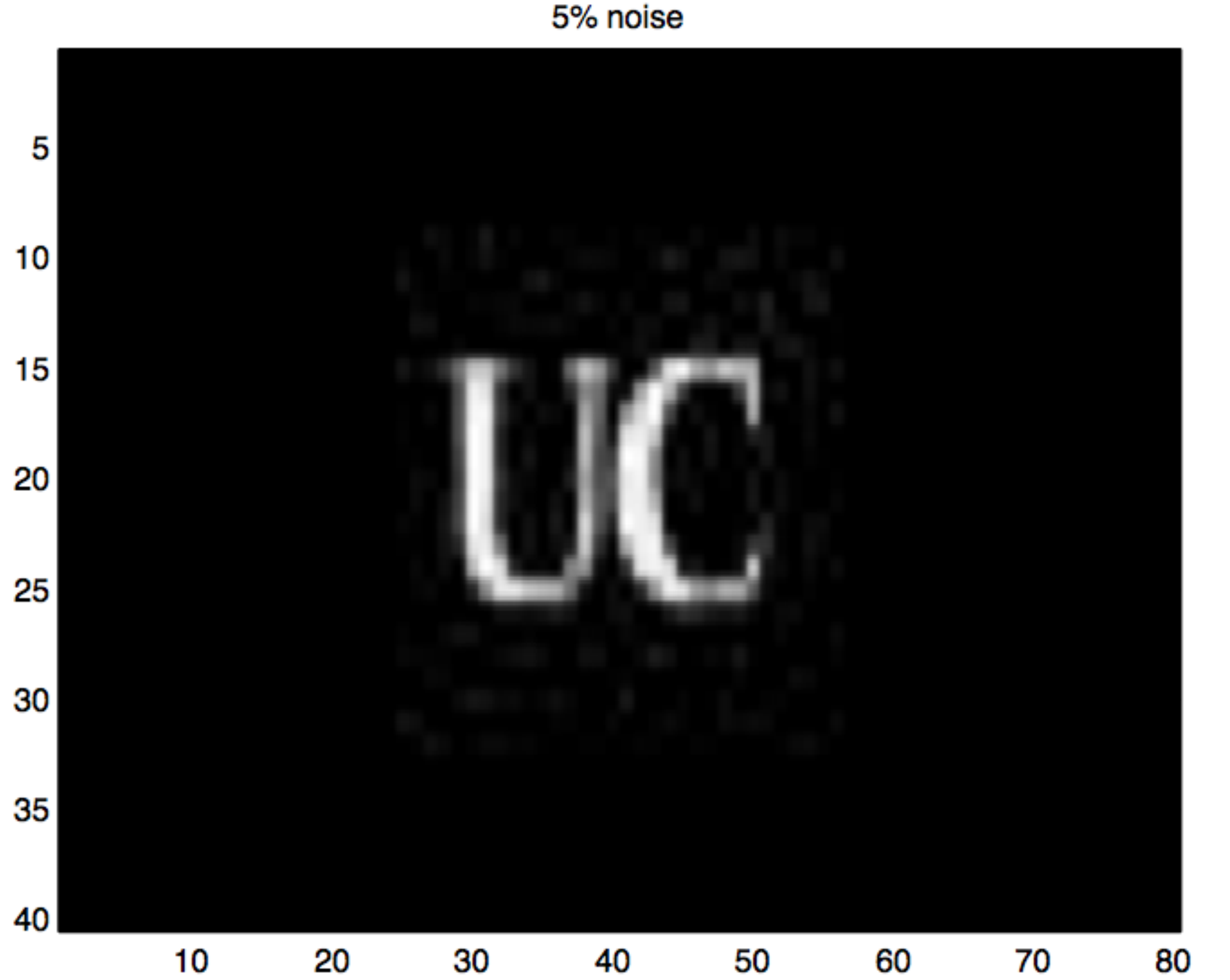}
\includegraphics[width=0.32\textwidth]{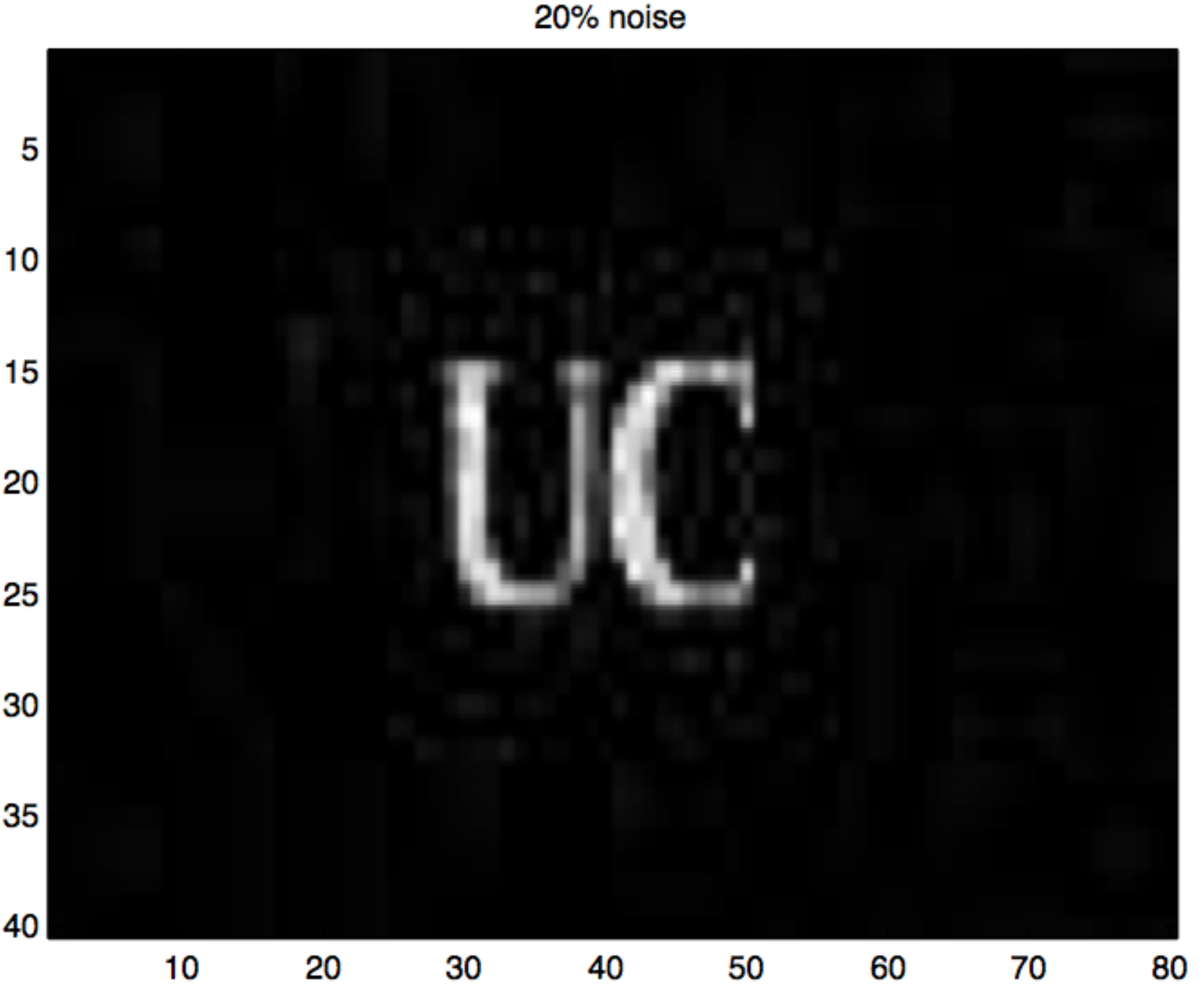}
\end{center}
\caption{The original $40\times 80$ pixel image (left) and the BPDN reconstructions (middle panel,  $5\%$ noise; right panel $20\%$ noise) with one sensor and 500 random illuminations. }
\label{fig8}
\end{figure}

 \begin{figure}[h]
\begin{center}
\includegraphics[width=0.32\textwidth]{UC0.pdf}
\includegraphics[width=0.32\textwidth]{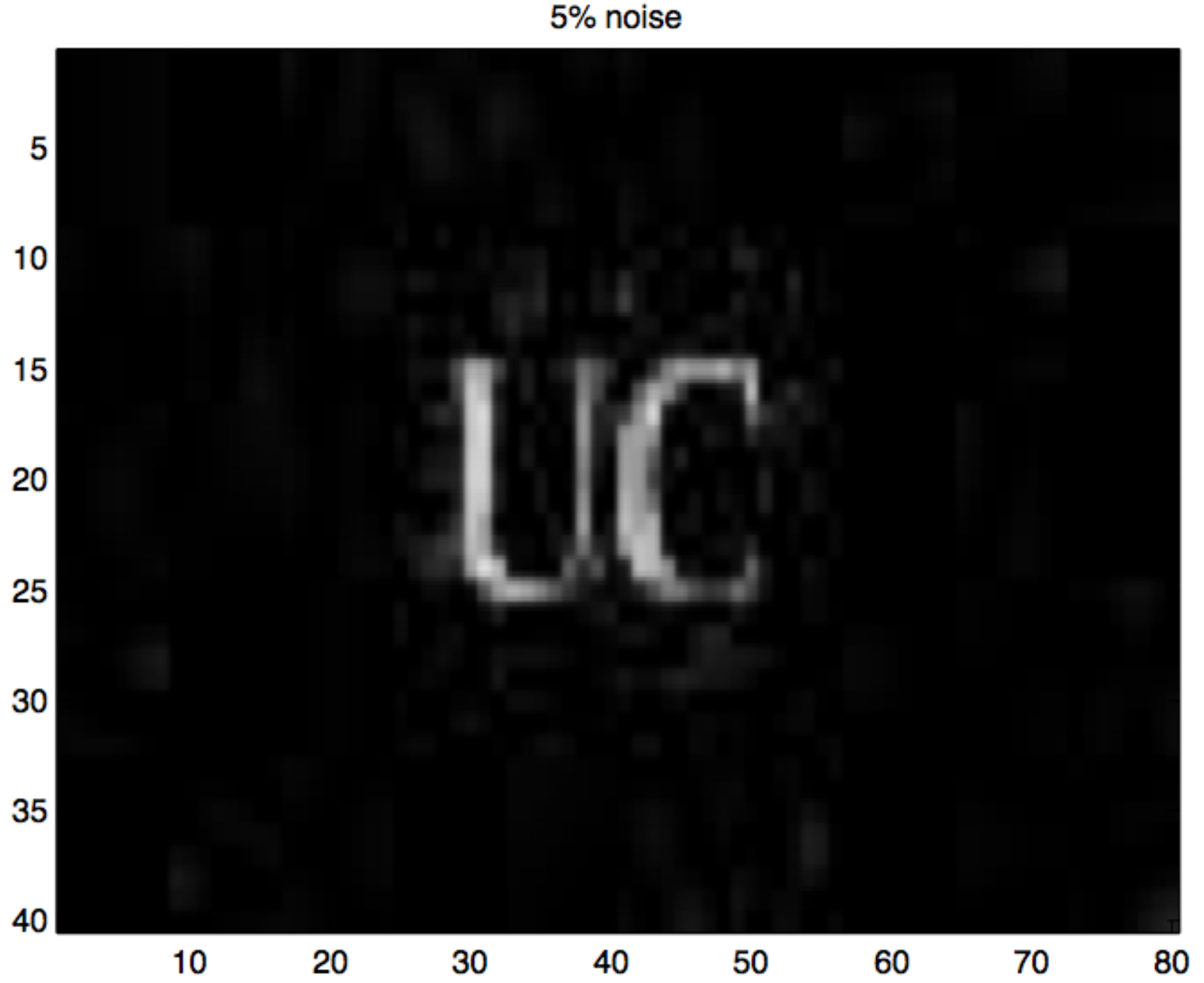}
\includegraphics[width=0.32\textwidth]{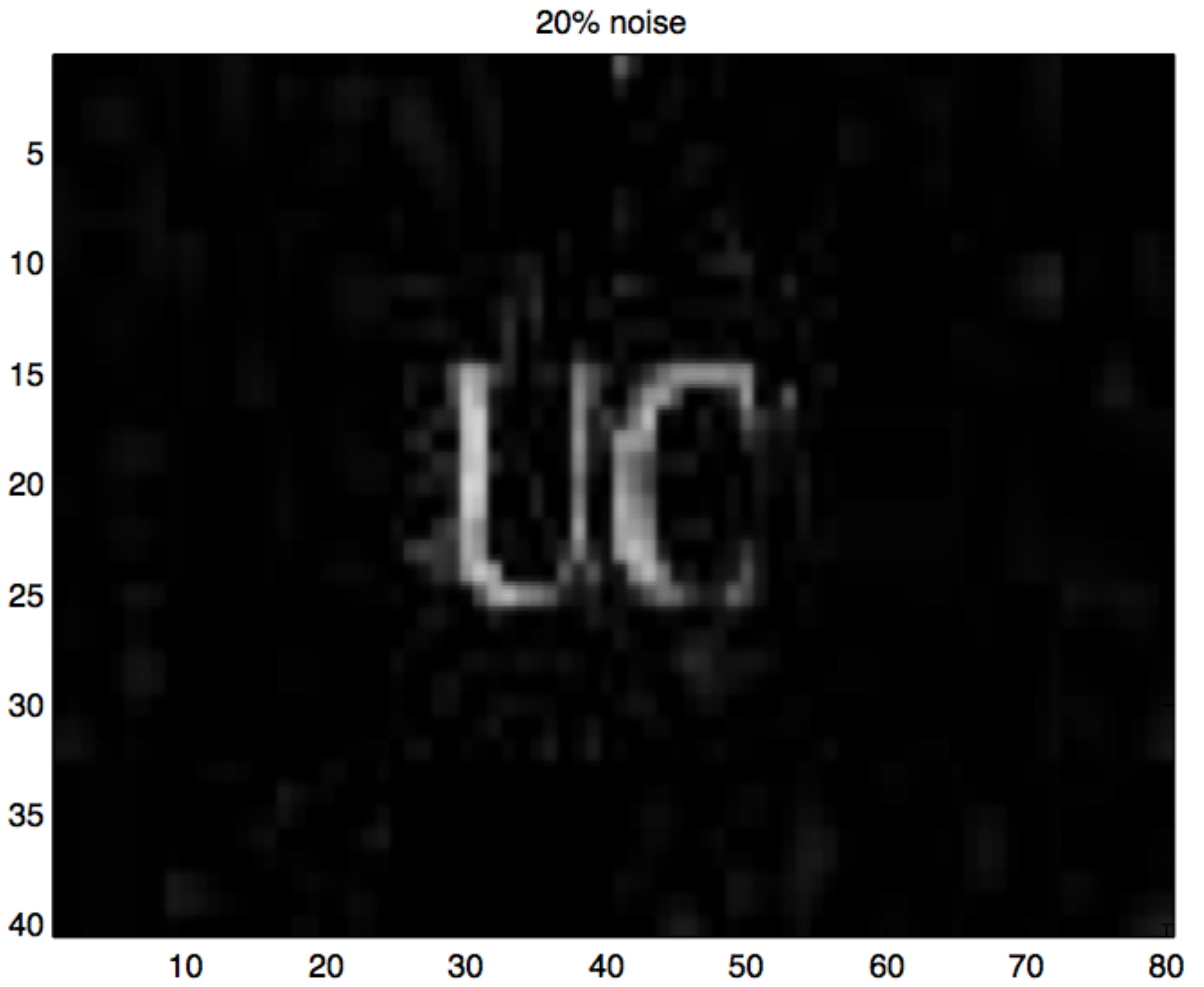}
\end{center}
\caption{The original $40\times 80$ pixel image (left) and the BPDN reconstructions (middle panel, $5\%$ noise; right panel, $20\%$ noise)  with one illumination and 500 randomly distributed sensors.  }
\label{fig8'}
\end{figure}

 \begin{figure}[h]
\begin{center}
\includegraphics[width=0.32\textwidth]{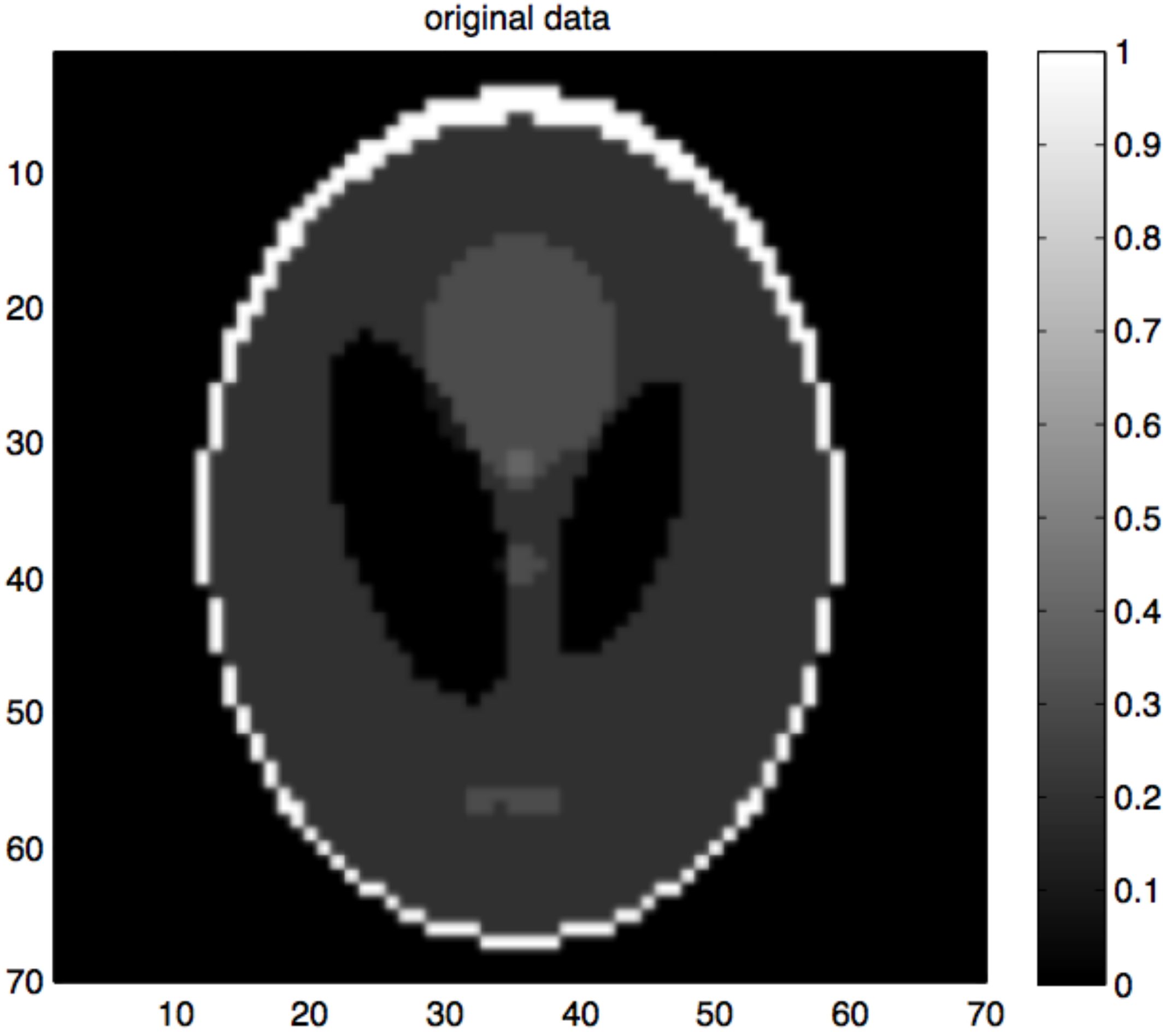}
\includegraphics[width=0.32\textwidth]{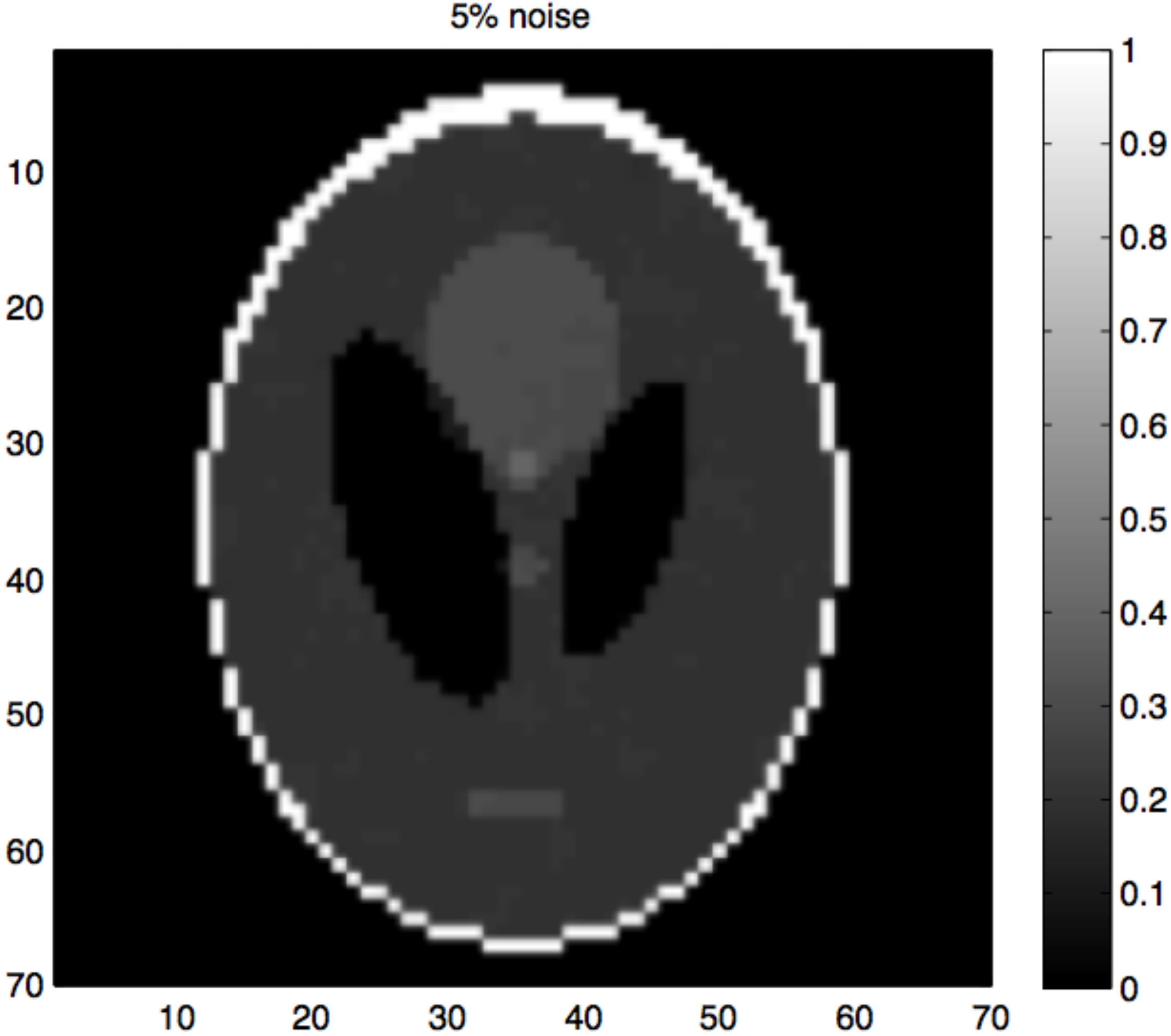}
\includegraphics[width=0.32\textwidth]{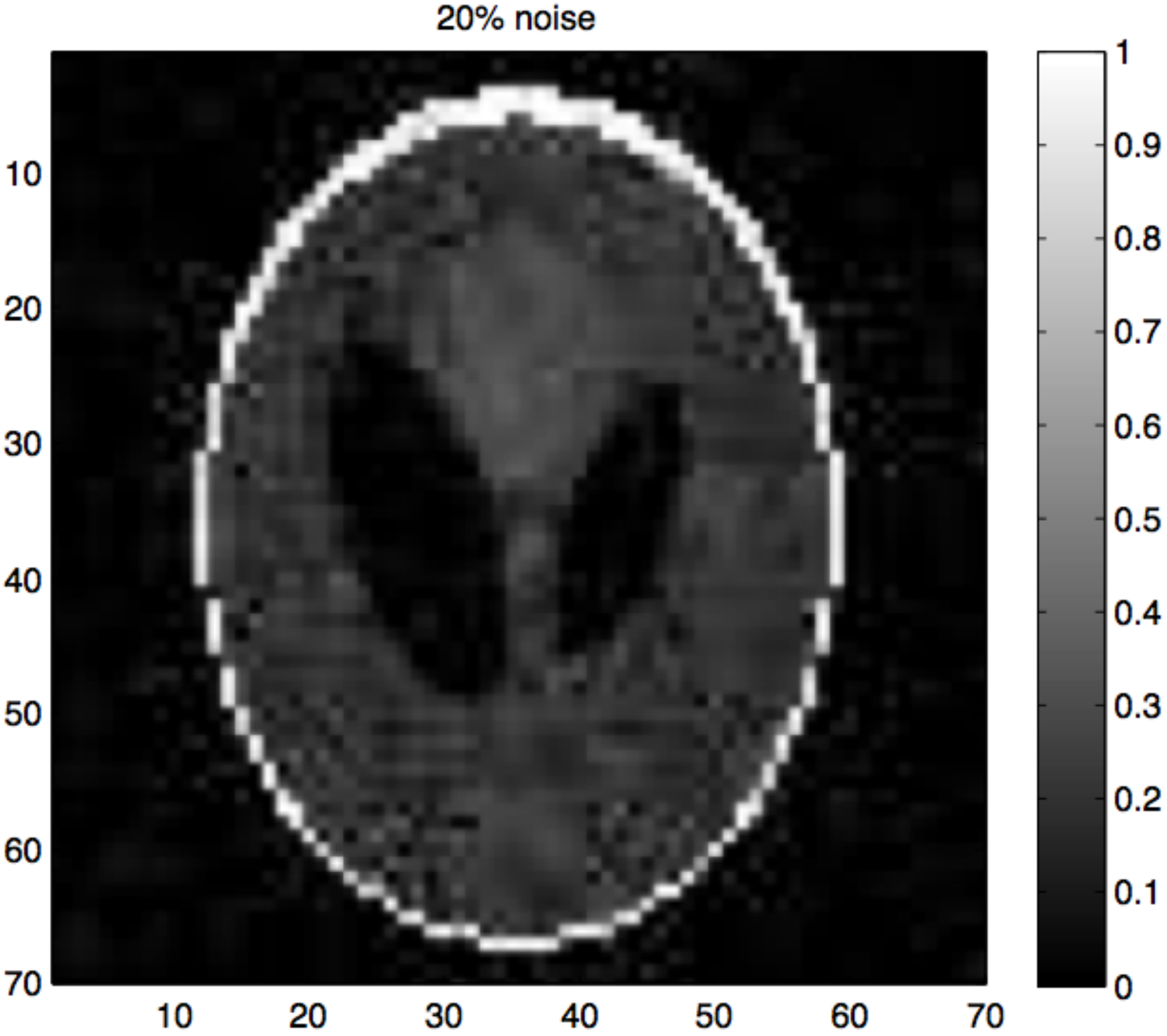}
\end{center}
\caption{The original $70\times 70$ pixel image (left), the Shepp-Logan phantom, and  the TV-minimization  reconstructions 
(middle panel, $5\%$ noise; right panel, $20\%$ noise) with
one sensor and 1000 random illuminations.  }
\label{fig9}
\end{figure}

 \begin{figure}[h]
\begin{center}
\includegraphics[width=0.32\textwidth]{SLP0.pdf}
\includegraphics[width=0.32\textwidth]{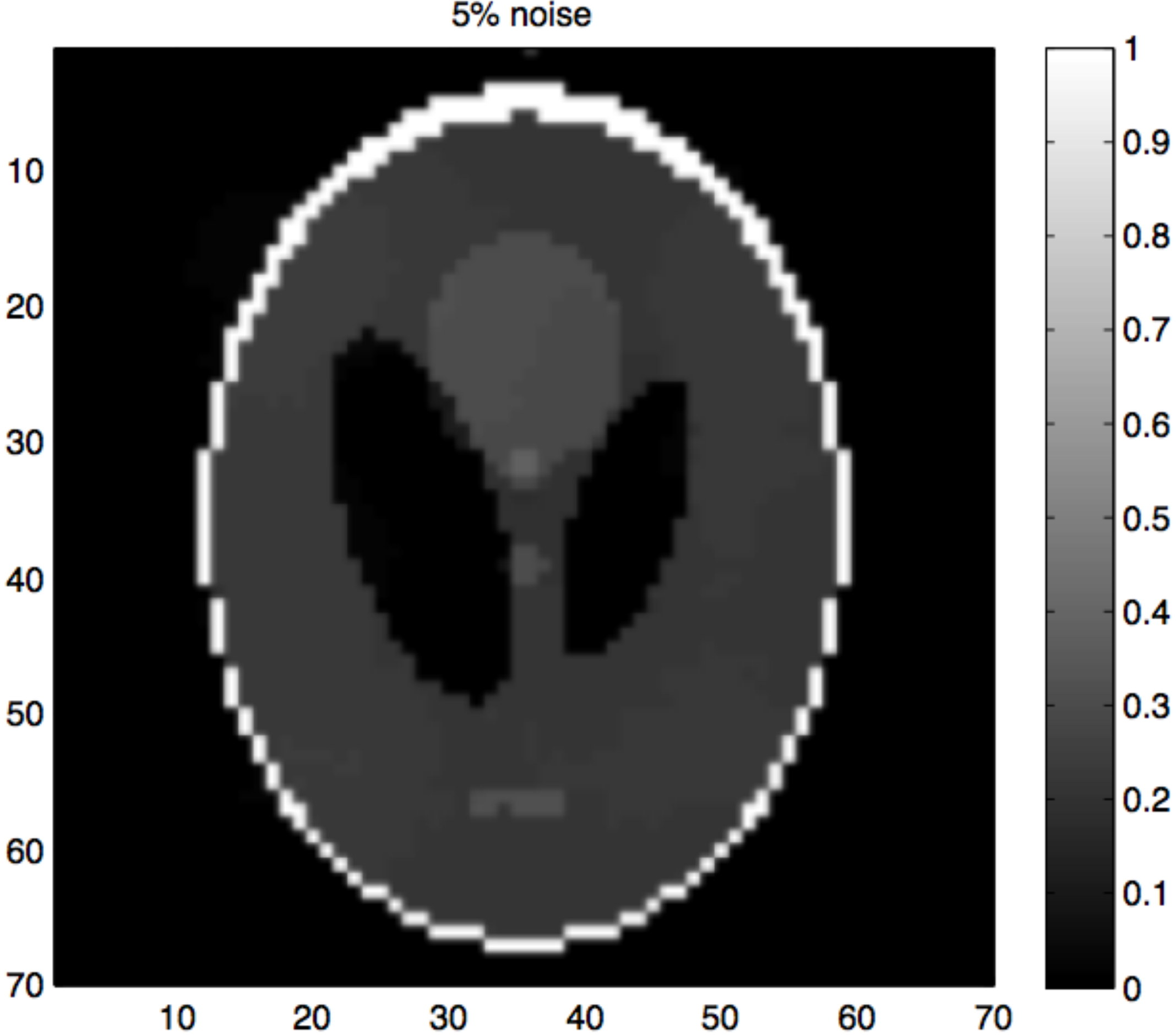}
\includegraphics[width=0.32\textwidth]{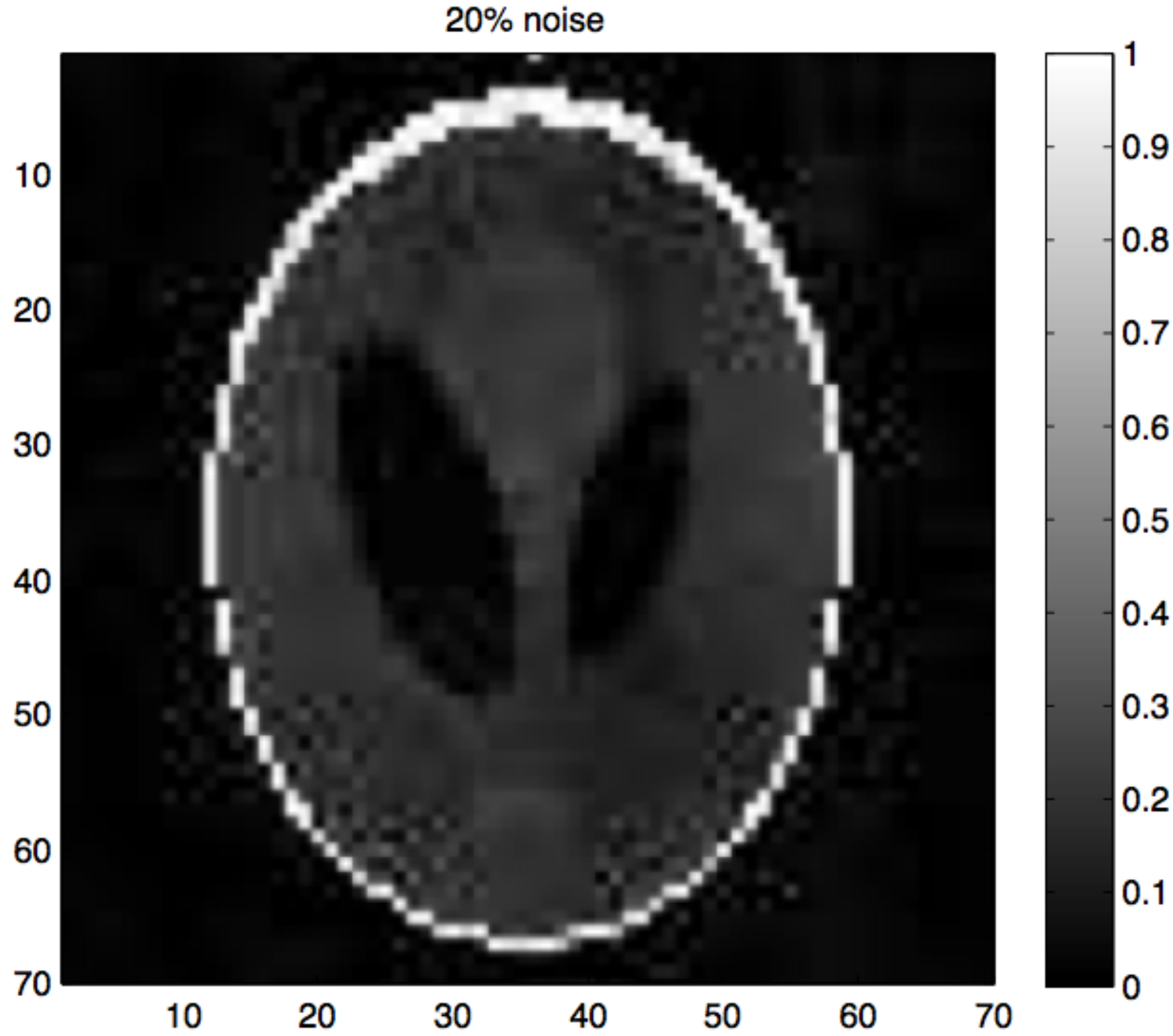}
\end{center}
\caption{The original $70\times 70$ pixel image (left), the Shepp-Logan phantom, and the TV-minimization  reconstructions (middle panel, $5\%$ noise; right panel, $20\%$ noise)  with one illumination 
and 1000 randomly distributed sensors.  }
\label{fig9'}
\end{figure}

 \section{Conclusion}\label{sec9}
 
 We have proposed a new approach to superresolving
 point and extended objects
 based on random illumination and compressed sensing
 reconstruction.
 
  We have proved that in the diffraction-limited case both the Lasso
 and the OST with random illumination can exactly localize
 $s=\cO(m)$ objects where  the number of data $m$
 is the product of the numbers of random probes  and
  sensors. For the under-resolved case where the
  Rayleigh resolution limit is broken, the Lasso
  still has a similar performance guarantee if
  the number of random illuminations  is sufficiently large. 
  It is possible to extend the OST result to the under-resolved
  case which is omitted here to simplify the presentation.

  Numerical evidence supports our theoretical prediction
  and confirms  the superiority of the Lasso to
  the OST in the set-up with random illumination. 
  
  We have also shown that  the BPDN is suitable for
  imaging extended objects and have provided numerical examples
  to demonstrate its performance.

  The superresolution effect with random illumination  revealed here 
  contrasts with  the subwavelength resolution  with deterministic near-field
  illumination  studied in \cite{subwave}.
  
  Finally we note that in our approach it is essential to measure
  the wave field. For {\em intensity-only} measurements, additional techniques such as interferometry or 
  phase retrieval methods are necessary for object reconstruction. 
  

\commentout{
\section{Piecewise smooth objects}
 In this section, we analyze the case of piecewise smooth objects.
 Instead of the objects themselves, the derivatives of such objects are sparse. 

 By (\ref{271}) and integration by parts we see that
 \beq
\label{272}
{\imath \om \over z_0}  (\cF O)_i \ba_l ={1\over  g(\ba_l)} \sum_{j=1}^N e^{\imath \theta_{kj}}   \int_{\square_j} \nabla O(x,y) e^{-\imath \om \xi_lx_j/z_0}e^{-\imath \om \eta_ly_j/z_0}  dx dy, 
 \eeq
 whose left hand side can be thought of as  the uncontaminated signal  produced by the sparse object $\nabla O$. 
 
Instead of  the vector-valued object  $\nabla O$ we employ
the complexified object 
 $I=\partial_x O+\imath \partial_y O$ provided that $O$ is real-valued.
Accordingly  we set the data as 
 \[
  {\imath \om \over z_0}  (\cF O)_i (\xi_l+\imath\eta_l),\quad i=(k-1)n+l.
  \]
  The sensing matrix is exactly the same as (\ref{ext1}). 
  
 Proceeding as before, we seek to first reconstruct the discrete approximation $X=(I(\br_j))\in \IC^{np}$ on the grid $\cL$
and recover the discrete approximation $O_\ell$ by integrating $I_\ell=\sum_{j} I(\br_j) \II_{\square_j}$.

We reconstruct $X$ by the BPDN (\ref{269})
subject to the additional curl-free constraint 
\beq
\label{277}
D_y\Re(Z)=D_x \Im(Z),\quad Z\in \IC^{np}
\eeq
where $D_x$ and $D_y$ are the discrete analogs of $\partial_x$
and $\partial_y$, respectively. 
BPDN with the additional constraint (\ref{277}) is equivalent to the the discrete version of the Total Variation Minimization (TVM) \cite{Cha, CL, ROF}. 

Instead of following the standard
TVM approach \cite{Cha}, we resort
to BPDN
to obtain a quantitative performance guarantee analogous
to Theorem \ref{thm3}. 

First, we observe the curl-free constraint (\ref{277}) defines
a convex subset, indeed a subspace over $\IR$,  in $\IC^{np}$. Hence BPDN with (\ref{277})
is a convex, quadratic  program that can be solved efficiently. 
Second the proof of Proposition \ref{prop3} (see \cite{music} for
the complex-value version) is valid without change under
the constraint (\ref{277}). 
In other words, Proposition \ref{prop3} holds true for BPDN with
the additional constraint (\ref{277}). Hence analogous to Theorem
\ref{thm3} we have the following statement. 

\begin{theorem}  
Suppose that 
\[
 {\|I-I_{\ell}\|_{L^1}}\leq {\ep\over \sqrt{np}}  \min_l |g(\ba_l)|
 \]
 holds true where $g$ is given by (\ref{260}).  Suppose that  the external noise $E_{\rm ext}$ satisfies $\|E_{\rm ext}\|_2\leq \ep'$ such that the total error $\|E\|_2\leq \|E_{\rm ext}\|_2+\|E_{\rm disc}\|_2\leq \ep+\ep'$. Then the reconstruction  $\hat X$ by BPDN with (\ref{277}) satisfies the error bound (\ref{280})
for all $s$ satisfying (\ref{321'}). 

\end{theorem}
}

  \bigskip
  
 {\bf Acknowledgement} \,\, I am grateful to  Mike Yan for producing Figures \ref{fig1}-\ref{fig6} and Hsiao-Chieh Tseng for
 producing Figures \ref{fig8}-\ref{fig9'} of Section \ref{sec:num}. 

\commentout{

}

\end{document}